\documentclass{article}

\usepackage{arxiv}          

\usepackage[utf8]{inputenc} 
\usepackage[T1]{fontenc}    
\usepackage[colorlinks=true,linkcolor=red,          
    citecolor=blue,         
    filecolor=magenta,      
    urlcolor=cyan]{hyperref}
\usepackage{url}            
\usepackage{booktabs}       
\usepackage{nicefrac}       
\usepackage{microtype}      
\usepackage{amsmath, amsthm, amssymb, amstext}
\usepackage{xcolor,tikz}    
\usepackage{algorithm2e}

\usetikzlibrary{calc}
\tikzstyle{NN}=[>=latex, every node/.style={scale=0.9}, every label/.style={scale=0.8}]
\tikzstyle{neuron}=[circle, fill=black!40, minimum size=10pt, inner sep=2pt]
\tikzstyle{circ}=[circle, fill=white, draw=black, minimum size=12pt, inner sep=2pt]
\tikzstyle{weight}=[circle, fill=white, inner sep=1pt]
\tikzstyle{edge}=[->, shorten <=2pt, shorten >=2pt, draw=black]

\newcounter{theorem}
\newtheorem{thm}[theorem]{Theorem}
\newtheorem*{definition}{Definition}

\renewcommand{\t}[1]{#1'} 
\newcommand{\E}{\operatorname{\mathbb{E}}}

\newcommand{\var}{\operatorname{Var}}

\newcommand{\Span}{\operatorname{span}}

\newcommand{\diag}{\operatorname{diag}}
\newcommand{\normal}{\mathcal{N}}
\newcommand{\argmin}{\operatorname*{arg\,min}}

\newcommand{\nnopg}{\texttt{NN\_OPG}}
\newcommand{\nnwrap}{\texttt{NN\_wrap}}

\newcommand{\nn}{\textsc{nn}}
\newcommand{\opg}{\textsc{opg}}
\newcommand{\mave}{\textsc{mave}}
\newcommand{\cve}{\textsc{cve}}
\newcommand{\ecve}{\textsc{ecve}}

\newcommand{\sir}{\textsc{sir}}
\newcommand{\pir}{\textsc{pir}}
\newcommand{\pfc}{\textsc{pfc}}
\newcommand{\contr}{\textsc{cr}}

\newcommand{\dr}{\textsc{dr}}

\newcommand{\save}{\textsc{save}}
\newcommand{\sdr}{\textsc{sdr}}

\newcommand{\phd}{\mbox{p{\textsc{h}d}}}

\newcommand{\gopg}{g_{\nnopg}}
\newcommand{\gwrap}{g_{\nnwrap}}
\newcommand{\gnn}{g_{\nn}}
\newcommand{\Loss}{\mathcal{L}}

\renewcommand{\epsilon}{\varepsilon}
\newcommand{\spc}{\mathcal{S}} 
\newcommand{\real}{\mathbb{R}} 

\newcommand{\1}{\mathbf{1}}

\newcommand{\Z}{{\mathbf Z}}
\newcommand{\X}{{\mathbf X}}
\newcommand{\I}{{\mathbf I}}
\newcommand{\Ub}{{\mathbf U}}

\newcommand{\x}{{\mathbf x}}
\newcommand{\bb}{{\mathbf b}}

\newcommand{\W}{{\mathbf W}}

\newcommand{\0}{{\mathbf 0}}
\newcommand{\B}{{\mathbf B}}

\newcommand{\V}{{\mathbf V}}
\newcommand{\eb}{\mathbf{e}}

\newcommand{\Pbf}{\mathbf{P}}

\newcommand{\vb}{\mathbf v}

\def\xn{{\mathbf x}}

\newcommand{\Pb}{\mathbb{ P}}

\newcommand{\greekbold}[1]{\mbox{\boldmath $#1$}}

\newcommand{\Sigmabf}{\greekbold{\Sigma}}
\newcommand{\Sigmaxbf}{\Sigmabf_{\x}}

\newcommand{\Thetabf}{\greekbold{\Theta}}

\newcommand{\ms}{\mathcal{S}_{\E\left(Y\mid\X\right)}} 

\title{Fusing Sufficient Dimension Reduction with  Neural Networks}

\author{\textbf{Daniel Kapla} \thanks{daniel.kapla@tuwien.ac.at} \\
	\hspace{.2cm}
	Institute of Statistics and Mathematical Methods in Economics\\ Faculty of Mathematics and Geoinformation\\ TU Wien, Vienna, Austria
	\And \\
	\textbf{Lukas Fertl}\thanks{lukas.fertl@tuwien.ac.at}\\
    \hspace{.2cm}
	Institute of Statistics and Mathematical Methods in Economics\\ Faculty of Mathematics and Geoinformation\\ TU Wien, Vienna, Austria
	\And \\
	\textbf{Efstathia Bura} \thanks{efstathia.bura@tuwien.ac.at} \\
	\hspace{.2cm}
	Institute of Statistics and Mathematical Methods in Economics\\ Faculty of Mathematics and Geoinformation\\ TU Wien, Vienna, Austria
	}

\begin{document}
\maketitle

\begin{abstract}
    We consider the regression problem where the dependence of the response $Y$
    on a set of predictors $\X$ is fully captured by the regression function
    $\E(Y \mid \X)=g(\B'\X)$, for an unknown function $g$ and low rank parameter
    $\B$ matrix. We combine neural networks with sufficient dimension reduction
    in order to remove the limitation of small $p$ and $n$ of the latter. We show
    in simulations that the proposed estimator is on par with competing sufficient
    dimension reduction methods in small $p$ and $n$ settings, such as
    \textit{minimum average variance estimation} and \textit{conditional variance
    estimation}. Among those, it is the only computationally applicable in large $p$ and $n$ problems. 
\end{abstract}

\section{Introduction}
	
In this paper we focus on the regression problem where the dependence of the response $Y$ on a set of predictors $\X$ is fully captured by the regression function $\E(Y \mid \X)$. Moreover, we further assume there exists a linear projection, or, \textit{reduction} of the predictors that encapsulates all the modeling information in $\X$ about $\E(Y \mid \X)$. 

Specifically, we assume the conditional distribution of $Y$ given a $p$-variate random vector $\X$ satisfies the regression model
\begin{equation}\label{mod:mean_basic}
    Y = g(\t{\B} \X) + \epsilon,
\end{equation}
where $\B \in \real^{p \times k}$ of rank $k < p$, $\epsilon  \in \real$ is a random variable with $\E(\epsilon \mid \X) = 0$ and  $\var(\epsilon)= \E\left(\epsilon^2\right)=\eta^2 < \infty$, and $g$ is an unknown continuously differentiable non-constant function. The projection $\t{\B}\X$ is a \textit{linear sufficient dimension reduction} since $\E(Y\mid \X)=\E(Y \mid \t{\B}\X)$. 
Ideally, $k =1,2$, or 3, in which case modeling $Y$ as a function of $\X$ is substantially simplified. 
Our goal is to find and estimate the linear projection $\t{\B}\X$ of $\X$ as accurately as possible.

The first method targeting the linear sufficient reduction in the general regression model $F(Y \mid \X)=F(Y \mid \t{\B}\X)$, where $F$ signifies the conditional cumulative distribution function of $Y$ given the conditioning argument,  was \textit{sliced inverse regression} (\sir, \cite{Li1991}). $\sir$, as well as most \textit{sufficient dimension reduction} ($\sdr$) methods, is based on the \textit{inverse regression} of $\X$ on the response  $ Y$. These include \textit{sliced average variance estimation} (\save, \cite{Cook2000}), \textit{parametric inverse regression} (\pir, \cite{BuraCook2001}), \textit{principal fitted components} (\pfc, \cite{CookForzani2008}), \textit{directional regression} (\dr, \cite{directionalRegression}), and \textit{contour regression} (\contr, \cite{contourRegression}). Further, there are model (likelihood) based sufficient dimension reduction methods, which require knowledge of the conditional or joint distribution and are researched in \cite{Cook2007,CookForzani2009,BuraForzani2015,BuraDuarteForzani2016}.
A recent overview of $\sdr$  methods can be found in \cite{Yin2010,MaZhu2013,Li2018}.

These methods require varying assumptions on either the joint distribution of $\t{(Y,\t{\X})}$, or the conditional distribution of $\X \mid Y$, limiting their applicability. A different approach focuses on the \textit{forward} regression of $Y$ on $\X$ in order to extract the reduction. The first such method, \textit{principal Hessian directions} ($\phd$), was introduced by \cite{Li1992} and was further developed by \cite{Cook1998b} and \cite{CookLi2002,CookLi2004}. \textit{Minimum average variance estimation} ($\mave$) was introduced by \cite{Xiaetal2002} and was generalized in \cite{Xia2007,WangXia2008}. \textit{Conditional variance estimation} ($\cve$, \cite{FertlBura2021a}) is the most recent addition to the forward regression  $\sdr$ methodology. These estimators require minimal assumptions on the smoothness of the joint distribution and frequently enjoy better estimation accuracy but at the expense of higher computational cost. Among those, the most prominent so far has been the \textit{minimum average variance estimation} (\mave) \cite{Xiaetal2002}. The recently developed \textit{conditional variance estimation} (\cve, \cite{FertlBura2021a}) and its generalization \textit{ensemble conditional variance estimation} (\ecve, \cite{FertlBura2021b}) has been shown to be the only approach among forward regression based $\sdr$ methods that exhibits on par or better performance than $\mave$.

All forward regression $\sdr$ methods, $\phd, \mave$ and $\cve$,  are usable in relatively small $p$  and $n$ regression problems. When both $p$ and $n$ increase substantially, their computation can spread over days or weeks, thus rendering them  infeasible in practice. Nowadays, many data applications easily exceed these thresholds. 

This paper combines forward regression $\sdr$ with neural networks, which excel in handling huge amounts of data, in order to remove the limitation of small $p$ and $n$. 
We propose a two stage $\nn-\sdr$ estimator that carries out simultaneous sufficient dimension reduction and neural network learning. 

First we fit an arbitrary neural net to the data, and in the second stage we refine the estimate with a specific architecture using a bottleneck. The premise of the two stage $\nn-\sdr$ estimator is conceptually similar to $\mave$ with the difference that we use neural nets as universal function approximators compared to nonparametric local linear smoothing methods. The advantage of this approach is that it retains the accuracy of  state of the art $\sdr$ methods while it can be easily deployed to large scale datasets frequently encountered in applications. It also  obtains predictions  at nearly no additional computational cost compared to fully nonparametric methods used in $\mave$ and $\cve$. 
Further, the extension of the proposed $\nn-\sdr$ estimator to online learning, where new data are dynamically added, is straightforward.

The paper is organized as follows. In Section~\ref{sec:meanSubspace} we give a short overview of the theoretical foundations of SDR, and in Section~\ref{sec:neural_networks_MLP} we present neural nets and the notation used throughout. In Section~\ref{sec:NN_estimator} we propose the novel two stage estimator and in Section~\ref{sec:algo} describe the algorithm. Then in Section~\ref{sec:MAVEandOPG} we draw the analogy to existing SDR methods and demonstrate its performance in Sections~\ref{sec:sim}, \ref{sec:bigData}, \ref{sec:DataAnalysis} via simulations and data examples. Our concluding remarks are in  Section~\ref{sec:discussion}.

\section{Mean Subspace}\label{sec:meanSubspace}
Let  $(\Omega ,{\mathcal {F}},\Pb)$ be a probability space, $Y$ a univariate continuous response and $\X$ a $p$-variate continuous predictor,  jointly distributed, with $\t{(Y,\t{\X})}:\Omega \to \real^{p+1}$.

The interest of this paper is to estimate the \textit{mean subspace} denoted by $\ms$ (see \cite{CookLi2002,CookLi2004,Li2018}).

\begin{definition}
A linear space $\ms \subseteq \real^p$ is called a \textit{mean subspace} if for any basis $\B \in \real^{p \times k}$ of $\ms$ with $\dim(\ms) = k \leq p$,
\begin{align}\label{meanSDR}
\E\left(Y \mid \X\right)&= \E(Y \mid \t{\B} \X)    
\end{align}
\end{definition}

An equivalent characterisation of $\ms$ is given by
\begin{align}\label{meanSDRprojection}
\E\left(Y \mid \X\right)&= E(Y \mid \Pbf_{\ms} \X)    
\end{align}
where $\Pbf_\spc$ is the orthogonal projection on the space $\spc$ with respect to the usual inner product. 
This can be seen by letting $\B \in \real^{p \times k}$ be a basis of $\ms$, then $\Pbf_{\ms}\X = \B (\t{\B} \B)^{-1} \t{\B}\X$. The only random element is $\t{\B}\X$ and therefore the sigma algebras generated by $\t{\B}\X$ and $\Pbf_{\ms}\X$ are the same. This yields the equivalence of \eqref{meanSDR} and \eqref{meanSDRprojection}.

For a predictor vector $\X$ whose density is supported on a  convex set, with positive definite variance-covariance matrix, $\var(\X)=\Sigmaxbf$,  the  \textit{mean subspace} model \eqref{meanSDR} is equivalent to the regression model  \eqref{mod:mean_basic} in the Introduction. 
Their equivalence derives from 
\begin{align}\label{mean_basic_proof1}
    \E\left(Y \mid \X \right) &= \E\left(g(\t{\B} \X) \mid \X \right) + \E\left(\epsilon \mid \X \right) \nonumber \\
                              &= g(\t{\B} \X) = \E\left(Y \mid \t{\B} \X \right).
\end{align}
so that $\Span\{\B\} = \ms$, where span is the column space of $\B$.

The \textit{mean subspace} $\ms$ captures all the information in $\X$ about $Y$ that is contained in the first conditional moment $\E(Y \mid \X)$. That is,  if we are only interested in the conditional mean, $\X$ can be replaced by $\t{\B} \X \in \real^k$ without any loss of information. When $k$ is significantly smaller than $p$, this results in substantial savings in reducing the complexity of the modeling problem.

Our focus in this paper is the estimation of the mean subspace. By the equivalence of \eqref{meanSDR} and \eqref{meanSDRprojection} the sufficient reduction depends only on the subspace $\ms$ and not on a particular basis. Therefore, without loss of generality, we let  $\B \in \spc(p, k)$, where
\begin{equation}\label{Smanifold}
    \spc(p, k) = \{\V \in \real^{p \times k}: \t{\V}\V = \I_k\},
\end{equation}
denotes the Stiefel manifold, that comprises of all $p \times k$ matrices with orthonormal columns.

\section{The Multi Layer Perceptron (MLP)}\label{sec:neural_networks_MLP}

In this section we briefly review the concept of a \emph{Multi Layer Perceptron} (\texttt{MLP} \cite{NN1,NN2,Goodfellow-et-al-2016,ISL}) and introduce the corresponding notation.

An \texttt{MLP} is the concatenation of layers. Each layer consists of simple functions $f^{(l)}(x) = \phi(\W^{(l)} \xn + \bb^{(l)})$, where $\W^{(l)}$ is a matrix of \textit{weights}, $\bb^{(l)}$ is the \textit{bias} vector of layer $l$, and together they form an affine transformation, on which the \textit{activation} function  $\phi(\cdot)$ is  applied component-wise. The formal definition is provided next. 

\begin{definition}\label{defn:MLP}
	A \emph{Multi Layer Perceptron} (\texttt{MLP}) with $N$ layers from $\real^p \to \real$ is a function with the following structure
	\begin{align}\label{MLP}
	     f_{\text{MLP}_N}(\xn; \Thetabf) &= f^{(N)}\circ f^{(N-1)} \circ ... \circ f^{(1)} (\xn)
\end{align}
	where $\Thetabf = (\W_1,\bb_1, \ldots, \W_N,\bb_N)$ and the $l$-th layer is given by
	\begin{displaymath}\label{layer}
		f^{(l)}(\xn; \W^{(l)}, \bb^{(l)}) = \phi^{(l)}( \W^{(l)}\xn + \bb^{(l)})
	\end{displaymath}
with weights $\W^{(l)}\in\mathbb{R}^{n_{l}\times n_{l-1}}$, bias $\bb^{(l)}\in\mathbb{R}^{n_l}$, and a non-constant, continuous 
activation function $\phi^{(l)}:\real \to \real$ that is applied component-wise.
\end{definition}

The notation $\Thetabf = (\W_1,\bb_1, \ldots, \W_N,\bb_N)$ means  that all parameters of an \texttt{MLP} are collected in vectorised form into the vector $\Thetabf= (\text{vec}(\W_1), \bb_1,\ldots, \text{vec}(\W_N), \bb_N) \in \real^{\sum_{j=1}^N n_{l-1}n_l +n_l}$, where the operation $\text{vec}:\real^{n_{l-1} \times n_l} \to \real^{n_{l-1}n_l}$ stacks the columns of a matrix one after another. Note that in general an \texttt{MLP} does allow for multi-dimensional output, but for the sake of simplicity we only allow univariate responses and therefore only univariate outputs.

The first layer that receives the input $\xn$ is called the \emph{input layer}, and the last layer is the  \emph{output layer}. All other layers are called \emph{hidden layers}. A widely used activation function is the so called \texttt{ReLU} (Rectified Linear Unit) given by
\begin{equation*}
    \phi_{\texttt{ReLU}}(x) = \max(0, x).
\end{equation*}
The \texttt{ReLu} activation function will be used throughout this paper. Other popular choices include sigmoid functions like the tangens-hyperbolicus.

Figure~\ref{fig:architecture_MLP} depicts a $3$ layer \texttt{MLP}, $f_{\text{MLP}_3}(\xn; \Thetabf)$, with input dimension $4$; i.e., $\xn = \t{(x_1,\ldots, x_4)} \in \real^4$. The first layer $f^{(1)}$ has output dimension $6$, or $6$ so called \textit{neurons}, $\W_1 \in \real^{6 \times 4}, \bb_1 \in \real^6$.  The second layer, $f^{(2)}$, has $4$ neurons with $\W_2 \in \real^{4\times 6}, \bb_2 \in \real^{4}$, and the output layer, $f^{(3)}$, has $1$ neuron  with $\W_3 \in \real^{1\times 4}, \bb_2 \in \real$. The arrows represent the weights of the layer. At each node (neuron), the bias is added before the activation function $\phi^{(l)}$ is applied.

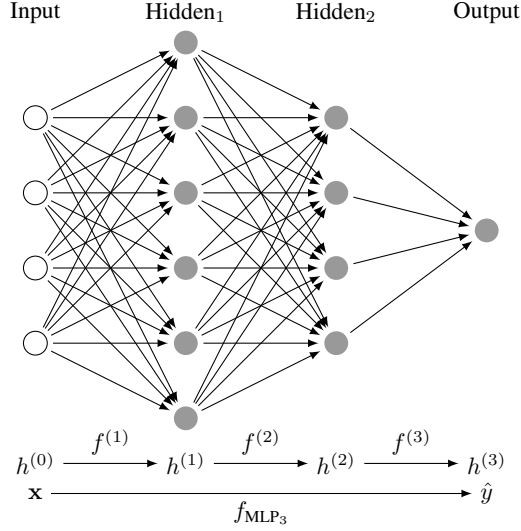
\begin{figure}[htbp!]
	\centering
	\begin{tikzpicture}[NN]
		\foreach \x in {1,...,4} {
			\node[neuron, fill=white, draw] (I-\x) at (1, -\x) {};
		}
		\foreach \x in {1,...,6} {
			\node[neuron] (H1-\x) at (3, {1-\x}) {};
		}
		\foreach \x in {1,2,3,4} {
			\node[neuron] (H2-\x) at (5, -\x) {};
		}
		\node[neuron] (O) at (7, -2.5) {};

		\foreach \from in {1,...,4} {
			\foreach \to in {1,...,6} {
				\draw[edge] (I-\from) -- (H1-\to);
			}
		}
		\foreach \from in {1,...,6} {
			\foreach \to in {1,...,4} {
				\draw[edge] (H1-\from) -- (H2-\to);
			}
		}
		\foreach \from in {1,...,4} {
			\draw[edge] (H2-\from) -- (O);
		}
		\node at ( 1, .4) {Input};
		\node at ( 3, .4) {Hidden${}_1$};
		\node at ( 5, .4) {Hidden${}_2$};
		\node at ( 7, .4) {Output};
		\node[anchor=south] at (2, -5.6) {$f^{(1)}$};
		\node[anchor=south] at (4, -5.6) {$f^{(2)}$};
		\node[anchor=south] at (6, -5.6) {$f^{(3)}$};
		\node (h_0) at (1, -5.6) {$h^{(0)}$};
		\node (h_1) at (3, -5.6) {$h^{(1)}$};
		\node (h_2) at (5, -5.6) {$h^{(2)}$};
		\node (h_3) at (7, -5.6) {$h^{(3)}$};
		\foreach \i/\j in {0/1,1/2,2/3} {
			\draw[->] (h_\i) -- (h_\j);
		}
		\node (in) at (1, -6) {$\xn$};
		\node (out) at (7, -6) {$\hat{y}$};
		\draw[->] (in) -- (out) node[anchor=north, pos=0.5] {$f_{\text{MLP}_3}$};
	\end{tikzpicture}
	\caption[MLP Architecture]{\label{fig:architecture_MLP}Example Architecture of a 3 Layer MLP.}
\end{figure}

The universal approximator theorem \cite[Thm  3]{HORNIK1991251} established that \emph{Multi Layer Perceptrons} (\texttt{MLP}s)   are universal approximators of functions. 
Theorem \ref{thm_universal_approximator}, which asserts that any continuously differentiable function can be approximated arbitrarily close on compact sets by an  \texttt{MLP}, reproduces it.
    
\begin{thm}\label{thm_universal_approximator}
    Let $\texttt{MLP}_\infty$ be the set of all one layer \texttt{MLP}'s with arbitrarily many neurons in the first layer and the activation function $\phi$ is non-constant and bounded, then $\texttt{MLP}_\infty$ is uniformly $m$ dense in $C^m(\real^p)$ on compact sets, where $C^m(\real^p)$ is the space of all $m$-times differentiable functions on $\real^p$.
\end{thm}

An application of Theorem~\ref{thm_universal_approximator} with $m=1$, yields that for $g(\t{\B}\xn) \in C^1(\real^p)$ of model~\eqref{mod:mean_basic}, for every arbitrary compact set $K \subset \real^p$ and for all $\nu > 0$, there exists a one layer \texttt{MLP} $f_{\text{MLP}_1}(\cdot; \Thetabf)$ such that
\begin{equation*}
    \sup_{\xn \in K}\left(|g(\t{\B}\xn)-f_{\text{MLP}_1}(\xn; \Thetabf)| + \|\nabla_\xn g(\t{\B}\xn)- \nabla_\xn f_{\text{MLP}_1}(\xn; \Thetabf)\| \right) \leq \nu
\end{equation*}
Therefore, the conditional expectation $\E(Y \mid \X = \xn) = g(\t{\B}\xn)$ and its gradients can be approximated arbitrarily close on compact sets by a one layer \texttt{MLP}. This serves as the basis for the proposed estimation procedure in Section~\ref{sec:NN_estimator}.

\section{Neural Net SDR}\label{sec:NN_estimator}
Theorems~\ref{thm_OPG} and \ref{MAVE_thm} present two ways of identifying $\B$ in model \eqref{mod:mean_basic} at the population level. They serve as the motivation for the proposed estimators. 

\begin{thm}\label{thm_OPG}
    Assume model~\eqref{mod:mean_basic} holds. Let
    \begin{align}
        \Sigmabf_\nabla = \E\left(\nabla \tilde{g}(\X)\nabla \t{\tilde{g}(\X)}\right) = \E\left(\nabla g(\t{\B}\X)\nabla \t{g(\t{\B}\X)}\right) \label{gradient_matrix}
    \end{align}
    where $\tilde{g}(\xn) = \E(Y\mid \X = \xn) = \E(Y\mid \t{\B}\xn) = g(\t{\B}\xn)$. Then, $\Span\{\Sigmabf_\nabla\} = \Span\{\B\}$.
\end{thm}
\begin{proof}
    Model~\eqref{mod:mean_basic} implies $\E(Y\mid \X = \xn) = \E(Y\mid \t{\B}\xn) = g(\t{\B}\xn)$. Differentiating yields
    \begin{align*}
        \nabla_\xn g(\t{\B}\xn) = \frac{\partial (\t{\B} \xn)}{\partial \xn} \nabla_{\t{\B}\xn} g(\t{\B}\xn) =\B \nabla g(\t{\B} \xn) 
    \end{align*}
    Therefore,  
    \begin{multline*}
        \Sigmabf_\nabla = \E\left(\nabla_\xn g(\t{\B}\X)\nabla_\xn \t{g(\t{\B}\X)}\right) \\
            = \B \E\left(\nabla g(\t{\B} \X) \nabla \t{g(\t{\B} \X)} \right) \t{\B}
    \end{multline*}
    which completes the proof.
\end{proof}

For $\V \in \spc(p,k)$, let 
\begin{align}
    T(\V) &= \E(Y - \E(Y\mid \t{\V}\X))^2 \nonumber \\
          &= \E\left(\E[(Y - \E(Y\mid \t{\V}\X))^2\mid \t{\V}\X]\right) \label{NN_target} \\
          &= \E(\sigma^2(\t{\V}\X)) \label{MAVE_target_function}
\end{align}
where $\sigma^2(\t{\V}\X) = \E[(Y - \E(Y\mid \t{\V}\X))^2\mid \t{\V}\X]$.
$T(\V)$ in \eqref{MAVE_target_function} is the target function at the population level for $\mave$ and identifies $\ms$ as shown in Theorem~\ref{MAVE_thm}.
 
\begin{thm}\label{MAVE_thm}
    Assume model~\eqref{mod:mean_basic} holds and $T(\V)$ is defined in \eqref{MAVE_target_function}.  Then, 
    \begin{align}
        \Span\{\B\} &= \Span\{ \argmin_{\V \in \spc(p,k)} T(\V)\}.
    \end{align}
\end{thm}
\begin{proof}
    By \eqref{mod:mean_basic}, $Y = g(\t{\B}\X) + \epsilon$. Then, 
    \begin{align}
        T(\V) &= \E\left(g(\t{\B}\X) - \E(Y\mid \t{\V}\X)\right)^2 \nonumber \\
              &\phantom{=} + 2 \E\left([g(\t{\B}\X) - \E(Y\mid \t{\V}\X)]\epsilon\right) + \var(\epsilon) \notag \\
              &= \E\left(g(\t{\B}\X) - \E(Y\mid \t{\V}\X)\right)^2 + \var(\epsilon) \geq \var(\epsilon) \label{MAVE1}
    \end{align}
    since
    \begin{multline*}
        \E\left([g(\t{\B}\X) - \E(Y\mid \t{\V}\X)]\epsilon\right) \\
            = \E\left([g(\t{\B}\X) - \E(Y\mid \t{\V}\X)]\E\left(\epsilon\mid \X \right)\right) = 0
    \end{multline*}
    due to $\E(\epsilon \mid \X) = 0$.

    For all $\V$ such that $\Span\{\V\} = \Span\{\B\}$, \eqref{MAVE1} yields
    \begin{displaymath}
        T(\V) = \var(\epsilon) = \eta^2
    \end{displaymath}
    For all $\V$ such that $\Span\{\V\} \neq \Span\{\B\}$, $\E(Y\mid \t{\V}\X) \neq g(\t{\B}\X)$, and
    \begin{displaymath}
        T(\V) = \E\left(g(\t{\B}\X) - \E(Y\mid \t{\V}\X)\right)^2 + \var(\epsilon) > \var(\epsilon) = \eta^2,
    \end{displaymath}
    which completes the proof.
\end{proof}

Next we define the three different neural nets we use in the proposed estimators.

\begin{definition}\label{defn:NNstages}
    Using the notation for MLP's in Section~\ref{sec:neural_networks_MLP}, we define
    \begin{align}
              \gopg(\xn; \Thetabf_1) &= f_{\text{MLP}_N}(\xn; \Thetabf_1): \real^p \to \real \label{NN_OPG_MLP}  \\
             \gwrap(\xn; \Thetabf_2) &= f_{\text{MLP}_N}(\xn; \Thetabf_2): \real^k \to \real \label{NN_wrap_MLP} \\
        \gnn(\xn; (\V, \Thetabf_2) ) &= \gwrap(\t{\V} \xn; \Thetabf_2):    \real^p \to \real \label{NN_MLP}
    \end{align}
    where $\V \in \spc(p,k)$ given in \eqref{Smanifold}.
\end{definition}

Our estimation method is run in two stages. The first uses $\gopg(\xn; \Thetabf_1)$ in \eqref{NN_OPG_MLP} as an estimator for $\tilde{g}$ in Theorem~\ref{thm_OPG} . The second, or,  \textit{refinement} stage, estimates $g$ in model~\eqref{mod:mean_basic}  with $\gwrap(\xn; \Thetabf_2)$ in \eqref{NN_wrap_MLP} and $g(\t{\B}\xn)$ in model~\eqref{mod:mean_basic} with $\gnn(\xn; (\V,\Thetabf_2) )$ in \eqref{NN_MLP}. Both \eqref{NN_OPG_MLP} and \eqref{NN_MLP} are used to estimate $\E(Y \mid \X = \xn)$ but the latter uses the specific structure of model~\eqref{mod:mean_basic} for refinement.

The MLP in \eqref{NN_MLP} is the same as in \eqref{NN_wrap_MLP} save for an additional input layer with the identity as the activation function; i.e., $\phi^{(1)}(x) = x$. Further, the first layer forms a bottleneck, as depicted in Figure~\ref{fig:architecture_MLP_bottleneck}, since $\t{\V} \xn \in \real^k$ with $k \ll p$.  \eqref{NN_OPG_MLP} serves as an estimate for $\E(Y \mid \X = \xn) = g(\t{\B}\xn): \real^p \to \real$, \eqref{NN_wrap_MLP} for $g:\real^k \to \real$, and \eqref{NN_MLP} is a refined estimate of $g(\t{\B}\xn)$ in model~\eqref{mod:mean_basic}. The bottleneck of the \texttt{MLP} in \eqref{NN_MLP} is conceptually similar to autoencoders [see, e.g.,  \cite{NNAutoencoder,NNAutoencoder2}] with the important difference  that the latter are analogous to nonlinear principal components and  unsupervised; that is,  independent of the response. 

Figure~\ref{fig:architecture_MLP_bottleneck} illustrates the MLP in \eqref{NN_MLP} with input dimension $p = 4$, $\xn = \t{(x_1,\ldots, x_4)} \in \real^4$. The first layer represents the 2-dimensional linear reduction $\V \in \spc(4,2)$ and the rest of the network coincides with \eqref{NN_wrap_MLP}.

For the proposed estimator, we can use any \texttt{MLP} that has more neurons than $p$ in the first hidden layer in \eqref{NN_OPG_MLP} and \eqref{NN_wrap_MLP}. For the sake of simplicity we opted for a $1$ layer \texttt{MLP} with $512$ neurons as default, since this gave satisfactory results in simulations. Further, the performance in simulations was robust against different architectures if sufficient regularisation is applied via \textit{dropout} in the training of the \texttt{MLP} [see \cite{dropout}].
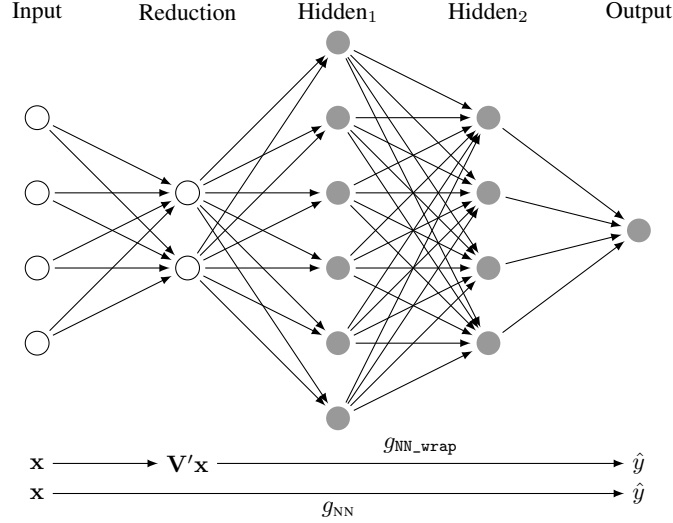
\begin{figure}[htbp!]
	\centering
	\begin{tikzpicture}[NN]
		\foreach \x in {1,...,4} {
			\node[neuron, fill=white, draw] (I-\x) at (-1, -\x) {};
		}
        \node[neuron, fill=white, draw] (R-1) at (1, -2) {};
        \node[neuron, fill=white, draw] (R-2) at (1, -3) {};
		\foreach \x in {1,...,6} {
			\node[neuron] (H1-\x) at (3, {1-\x}) {};
		}
		\foreach \x in {1,2,3,4} {
			\node[neuron] (H2-\x) at (5, -\x) {};
		}
		\node[neuron] (O) at (7, -2.5) {};

		\foreach \from in {1,...,4} {
			\foreach \to in {1,...,2} {
				\draw[edge] (I-\from) -- (R-\to);
			}
		}
		\foreach \from in {1,...,2} {
			\foreach \to in {1,...,6} {
				\draw[edge] (R-\from) -- (H1-\to);
			}
		}
		\foreach \from in {1,...,6} {
			\foreach \to in {1,...,4} {
				\draw[edge] (H1-\from) -- (H2-\to);
			}
		}
		\foreach \from in {1,...,4} {
			\draw[edge] (H2-\from) -- (O);
		}
		\node at (-1, .4) {Input};
		\node at ( 1, .4) {Reduction};
		\node at ( 3, .4) {Hidden${}_1$};
		\node at ( 5, .4) {Hidden${}_2$};
		\node at ( 7, .4) {Output};
        \node (h_0) at (-1, -5.6) {$\xn$};
        \node (h_1) at ( 1, -5.6) {$\t{\V} \xn$};
        \node (h_3) at ( 7, -5.6) {$\hat{y}$};
        \draw[->] (h_0) -- (h_1);
        \draw[->] (h_1) -- (h_3) node[pos=0.5, anchor=south] {$\gwrap$};
        \node (in)  at (-1, -6) {$\xn$};
        \node (out) at ( 7, -6) {$\hat{y}$};
        \draw[->] (in) -- (out) node[pos=0.5, anchor=north] {$\gnn$};
    \end{tikzpicture}
    \caption{\label{fig:architecture_MLP_bottleneck}{Illustration of $\gnn$ in \eqref{NN_MLP}}}
\end{figure}

\subsection{Initial Estimator}

We assume  $(Y_i,\t{\X}_i)_{i=1,\ldots,n}$ is a random sample from the joint distribution of $Y$ and $\X$ given by model~\eqref{mod:mean_basic}. Let 
\begin{align}\label{NN_OPG_target}
    T_\nnopg( \Thetabf_1) = \frac{1}{n} \sum_{i=1}^n \Loss\left(Y_i,\gopg(\X_i; \Thetabf_1)\right) 
\end{align}
be the objective function for the initial estimator, where $\Loss: \real \times \real \to [0,\infty)$ is a loss function. The training of the initial MLP in \eqref{NN_OPG_MLP} is carried out by minimizing the objective function in \eqref{NN_OPG_target}, 
\begin{align}\label{NN_OPG_training}
    \widehat{\Thetabf}_1 = \argmin_{\Thetabf_1}  T_\nnopg( \Thetabf_1)  
\end{align}
The resulting $\gopg(\xn;\widehat{\Thetabf}_1)$ is an estimate of $\E(Y \mid \X = \xn) = g(\t{\B}\xn)$ in model~\eqref{mod:mean_basic} if the squared error loss, 
\begin{align}\label{sqloss}
    \Loss(x,y) &= \left(x - y\right)^2,
\end{align}
is used.

We  set $\bb_i = \nabla_\xn \gopg(\X_i,\widehat{\Thetabf}_1) \in \real^p$ where $\widehat{\Thetabf}_1$ is defined in \eqref{NN_OPG_training}, which is an estimate for $\nabla g(\t{\B}\X_i) \in \real^p$. We let
\begin{align}\label{NN_OPG_matrix}
    \widehat{\Sigmabf}_{\nnopg} = \frac{1}{n}\sum_{j=1}^n \bb_j\t{\bb_j} \in \real^{p \times p}
\end{align}
that is an estimator for $\Sigmabf_\nabla$ in  \eqref{gradient_matrix}. The $\nnopg$ estimator is defined as 
\begin{align}\label{NN_OPG_estimator}
   \widehat{\B}_{\nnopg} = (\vb_1,\ldots,\vb_k) \in \spc(p,k)
\end{align}
where  $\vb_1,\ldots,\vb_k$ are the first $k$ eigenvectors of \eqref{NN_OPG_matrix}. 

By Theorem~\ref{thm_OPG}, under model~\eqref{mod:mean_basic}, 
$\Span\{\Sigmabf_\nabla\} = \Span\{\B\} = \ms$. If we assume that \eqref{NN_OPG_matrix} is a consistent estimator for \eqref{gradient_matrix}, then the $\nnopg$ estimator in \eqref{NN_OPG_estimator} is  consistent  for $\ms$ in model~\eqref{mod:mean_basic}. $\widehat{\B}_{\nnopg}$ in \eqref{NN_OPG_estimator} is used as an initial starting value for the optimization in \eqref{refinment_training} in order to obtain the refined estimator $\widehat{\B}_{\nn}$.

The loss function $\Loss$  is determined by model~\eqref{mod:mean_basic} and the conditional distribution of $Y \mid \X$. If the response $Y$ and predictors $\X$ are continuous and the error term in \eqref{mod:mean_basic} has a conditional Gaussian distribution, then the squared error loss function corresponds to the likelihood function. If $Y$ is Bernoulli  or multinomial distributed, then the cross entropy loss function can be used, and if $Y$ is Poisson distributed then the deviance is the natural choice for the loss function. In general, the loss function is the relevant part of the likelihood in the conditional distribution of $Y \mid \X$ and agrees with the loss function in generalized linear models for conditional distributions in the exponential family. 

\subsection{Refinement Estimator}
The second stage is the refinement of the initial estimator in \eqref{NN_OPG_estimator}. The $\nnopg$ estimator is obtained via the gradient of the trained MLP in \eqref{NN_OPG_MLP}. The training of the function $\gopg:\real^p \to \real$ in \eqref{NN_OPG_training} suffers from the curse of dimensionality if the input dimension is large. In this  case, the accuracy of the estimation of \eqref{NN_OPG_matrix} is adversely affected as learning a nonlinear function and its gradient with a high dimensional input space is difficult. The 
refinement procedure  explicitly incorporates the defining assumption of model~\eqref{mod:mean_basic} that a lower dimension projection of the input, $\t{\B}\X$, can replace the original input $\X$. This is realised via the function $\gnn$ in \eqref{NN_MLP}.

\begin{definition}
The target function for the refinement estimator is given by
\begin{align}\label{refinment_target_estimate}
    T_{\nn}(\V, \Thetabf_2) = \frac{1}{n} \sum_{i=1}^n \Loss\left(Y_i,\gnn(\X_i; (\V, \Thetabf_2))\right)
\end{align}
where $\Loss: \real \times \real \to [0,\infty)$ is a loss function, and $\V \in \spc(p,k)$. Further, we set 
\begin{align}\label{refinment_training}
    (\widehat{\B}_{\nn}, \widehat{\Thetabf}_2) = \argmin_{\V \in \spc(p,k), \Thetabf_2} T_{\nn}(\V, \Thetabf_2)
\end{align}
and  the $\nn$ refinement estimator is given by $\widehat{\B}_{\nn}$.
\end{definition}
The simultaneous optimization with respect to $\V$ and $ \Thetabf_2$ in \eqref{refinment_training} corresponds to simultaneous estimation of the sufficient reduction $\B$ and the link function $g$ in model~\eqref{mod:mean_basic}. The partially trained function $T_{\nn}(\cdot, \widehat{\Thetabf}_2)$ is an estimate for \eqref{NN_target} if the squared error loss function \eqref{sqloss} 
is used. 

Under squared error loss, if \eqref{NN_wrap_MLP}
is a consistent estimator for $g$ in \eqref{mod:mean_basic}, then $T_{\nn}(\cdot, \widehat{\Thetabf}_2)$ is  consistent  for $T(\V)$ in \eqref{NN_target}. By Theorem~\ref{MAVE_thm}, $\ms = \Span\{\B\} = \Span\{\argmin_{\V \in \spc(p,k)} T(\V)\}$ and the expectation would be that, subject to regularity conditions, the $\nn$ refinement estimator is consistent.

Nevertheless, proving the consistency of \eqref{NN_OPG_estimator} and the refinement estimator $\widehat{\B}_{\nn}$ in \eqref{refinment_training} requires  neural nets consistently estimate any function $g$ from a sample of model~\eqref{mod:mean_basic}. To the best of the authors' knowledge there is no such result available in the literature.

The optimization in \eqref{refinment_training} is solved via stochastic gradient descent training [see Section~\ref{sec:algo} or any other first order training algorithm for neural nets]. These algorithms require a starting value for the parameters, $(\V, \Thetabf_2)$, to be trained. In simulations the accuracy of the refined estimate $\widehat{\B}_{\nn}$ in \eqref{refinment_training} was very sensitive to the initialization of $\V$. 
We conjecture that a consistent estimator for $\B$ in \eqref{mod:mean_basic}, such as \eqref{NN_OPG_estimator}, is required in order to obtain a consistent estimate from the refinement procedure in \eqref{refinment_training}.

\section{Algorithm}\label{sec:algo}
In this section the computation  algorithm of the estimates $\widehat{\B}_{\nnopg}$ in \eqref{NN_OPG_estimator} and $\widehat{\B}_{\nn}$ in \eqref{refinment_training} is given. Both estimators depend on training a \texttt{MLP} using \texttt{tensorflow} \cite{tensorflow2015-whitepaper} with an \texttt{R} interface provided by the \texttt{R}-package \cite{tensorflow}.

For training the neuronal networks we use the \texttt{RMSProp} \cite{hinton2012RMSprop} algorithm, which is a variant of the (mini-batch) \textit{stochastic gradient descent} (SGD) algorithm \cite{Bottou1998SGD} (see, also, \cite{Goodfellow-et-al-2016} and \cite{tensorflow}). 

For regularisation during the training, we apply \texttt{dropout} with a rate of $0.4$ (see \cite{dropout}) after each fully connected hidden layer; i.e.,  the nodes are randomly set to $0$ with probability $0.4$ during each update step in the training procedure.

For a sample $\t{(Y_i, \t{\X}_i)}_{i=1,\ldots,n}$, we fix the \texttt{batch\_size} to $m \leq n$, and the \texttt{number of epochs} to $ep$. Let $ f_{\text{MLP}_N}(\xn; \Thetabf)$ be an $N$-layer \texttt{MLP}. The objective function is given by
\begin{equation*}
    T(\Thetabf) = \frac{1}{n} \sum_{i=1}^n \Loss\left(Y_i, f_{\text{MLP}_N}(\X_i; \Thetabf) \right)
\end{equation*}

A rough outline of stochastic gradient descent (\texttt{SGD}) is given in  Algorithm \ref{algo1}. 

\begin{algorithm}[ht]\label{algo1}
\SetAlgoLined
\KwResult{$\widehat{\Thetabf}^{(\text{end})}= \argmin_{\Thetabf} T(\Thetabf)$}
 Initialize:
 $\Thetabf^{(0)} $\\
 \For{ $u \in \{1,\ldots,ep\}$}{
 \For{$j \in \{1,\ldots, \lfloor n/m \rfloor \}$}{
        Determine the step sizes $\tau \in \real^{\dim(\Thetabf)}$ by \texttt{RMSProp}
        $\Thetabf^{(k+1)} = \Thetabf^{(k)} + \diag(\tau)  \sum_{l = (j-1)m+1}^{\min(jm, n)} \nabla_{\Thetabf} \Loss(Y_l, f_{\text{MLP}_N}(\X_l; \Thetabf^{(k)}))$
 }
Shuffle the dataset $\t{(Y_i, \t{\X}_i)}_{i=1,\ldots,n}$ randomly
}
\caption{Stochastic gradient descent outline}
\end{algorithm}


If the sample size $n$ is not a multiple of the batch size $m$, then in the last run of the inner loop the sum of gradients is extended to $n$. Further, if there are restrictions placed on some of the parameters,  the corresponding part of the parameter vector $\Thetabf$ is projected back to the restricted set after applying the update step. This is the case for $\V \in \spc(p,k)$ in \eqref{NN_MLP}, where the projection back onto the Stiefel manifold in \eqref{Smanifold} is done via a \textit{polar decomposition}.

An important feature of stochastic gradient descent training is that the complexity is linear in the sample size $n$ if the number of epochs $ep$ and the batch size $m$ are chosen independently of $n$. Moreover, in simulations we observed  that fewer epochs $ep$ suffice to find well trained neural nets for large samples. 

The proposed $\nn$ estimator for the \textit{mean subspace} is a two stage procedure, as described next.
\begin{description}
    \item[Stage 1:] Obtain the $\widehat{\B}_{\nnopg}$ estimator in \eqref{NN_OPG_estimator} by solving the optimization in \eqref{NN_OPG_training}. This is done via the stochastic gradient descent (SGD) algorithm with random initialization of the starting value $\Thetabf^{(0)}_1$ to obtain the estimate $\widehat{\Thetabf}_1 = (\widehat{\W}^{\nnopg}_1, \widehat{\bb}^{\nnopg}_1)$ in \eqref{NN_OPG_training}.
    \item[Stage 2:] Solve the optimization in \eqref{refinment_training} via the stochastic gradient descent (SGD) algorithm. The initial parameters of $\gnn$ in \eqref{NN_MLP} are set to $(\widehat{\B}_{\nnopg} ,\Thetabf_2^{(0)})$ where $\Thetabf_2^{(0)} = ( \widehat{\W}^{\nnopg}_1 \widehat{\B}_{\nnopg}, \widehat{\bb}^{\nnopg}_1)$. 
\end{description}
In the second stage, we use the weights and bias obtained by training $\gopg$ in  $\eqref{NN_OPG_MLP}$ as initialization for the parameters of $\gwrap$ in \eqref{NN_wrap_MLP}, and $\widehat{\B}_{\nnopg}$ as initial value for $\V \in \spc(p,k)$ in \eqref{NN_MLP}. 

This two stage initialization scheme is important for the performance of the proposed estimator since a random initialization of the parameters of the second stage adversely affects the accuracy of the estimator. 


\section{Analogy to Minimum Average Variance Estimation (MAVE) and Outer Product Gradient (OPG)} \label{sec:MAVEandOPG}

In this section we draw the analogy of $\opg$ and $\mave$, both introduced by \cite{Xiaetal2002}, with the proposed estimators.
The theoretical motivation for $\opg$ is given by Theorem~\ref{thm_OPG} and for $\mave$  by Theorem~\ref{MAVE_thm}. First we start by describing the $\mave$ algorithm briefly.

For the estimation of $\ms = \Span\{\B\}$ from an i.i.d. sample $(Y_i,\X_i)_{i=1}^n$ of model~\eqref{mod:mean_basic} we replace the target function $T(\V)$ in Theorem~\ref{MAVE_thm} by an estimate $T_n(\V)$, where $\V \in \spc(p,k)$. A local linear expansion of $\E(Y_i\mid \t{\V}\X_i)$ around $\X_0$ yields
\begin{align}\label{loclinear_approx}
    \E(Y_i\mid \t{\V}\X_i) \approx a + \t{\bb} \t{\V}(\X_i - \X_0),
\end{align}
where $a = g(\t{\V}\X_0) \in \real$, $\bb = \nabla g(\t{\V}\X_0) \in \real^k$. Therefore we obtain the following approximation for $\sigma^2(\t{\V}\X_0)$ in \eqref{MAVE_target_function},
\begin{multline*}
    \sigma^2(\t{\V}\X_0) \approx \sum_{i=1}^n (Y_i - \E(Y\mid \t{\V}\X))^2 w_{i,0} \\
                         \approx \sum_{i=1}^n (Y_i - a - \t{\bb} \t{\V}(\X_i - \X_0))^2w_{i,0}
\end{multline*}
for some weights $w_{i,0}$ that sum to 1 ($\sum_i w_{i,0} = 1$). 
The weights  
play a crucial role in the estimation. They are  given by
\begin{equation}\label{MAVE_weight}
    w_{i,0}(\V) := \frac{K\left(\frac{\t{\V}(\X_i-\X_0)}{h}\right)}{\sum_l K\left(\frac{\t{\V}(\X_l-\X_0)}{h}\right)},
\end{equation}
for a $k$ dimensional kernel $K(\cdot)$, and a bandwidth $h \in \real_+$. 

Since $K(\cdot) = \tilde{K}(\|\cdot\|_2)$ for a monotone decreasing univariate kernel $\tilde{K}(\cdot):\real \to \real_+$, \cite{Xiaetal2002} recognized that the weights depend only on the distance of $\t{\V}(\X_l-\X_0)$ (the further $\t{\V}\X_i$ is away from $\t{\V}\X_0$ the worse the linear expansion is and the less weight we assign). 

Then, an estimator for $\sigma^2(\t{\V}\X_0)$ in \eqref{MAVE_target_function} is given by
\begin{equation}\label{sigmahat}
    \hat{\sigma}^2(\t{\V}\X_0) := \min_{a,\bb} \sum_{i=1}^n  (Y_i - a - \t{\bb} \t{\V}(\X_i - \X_0))^2w_{i,0}(\V)
\end{equation}
and the target function $T(\V)$ is estimated by
\begin{multline}\label{MAVE_target_func}
    T_n(\V) = \frac{1}{n} \sum_{j=1}^n \hat{\sigma}^2(\t{\V}\X_j) \\
            = \min_{a_j,\bb_j} \frac{1}{n} \sum_{j=1}^n \sum_{i=1}^n (Y_i - a_j - \t{\bb_j}\t{\V}(\X_i - \X_j))^2 w_{i,j}(\V)
\end{multline}
where the weights are given in \eqref{MAVE_weight}. The Gaussian kernel is usually used with  bandwidth satisfying $h = h_n \propto n^{-1/(4+k)}$, as typically done in nonparametric function estimation, in order to obtain optimal asymptotic properties. 
 
\begin{definition}
    The $\mave$ estimator for $\ms = \Span\{\B\}$ in model \eqref{mod:mean_basic} is given by
    \begin{align}\label{MAVE_estimator}
        \widehat{\B}_{\mave} := \argmin_{\V \in  \spc(p,k)}T_n(\V).
    \end{align}
\end{definition}

\subsection{Analogy of \texorpdfstring{$\nn-\sdr$}{NNSDR} estimation to \texorpdfstring{$\mave$}{MAVE}}
The optimization 
in \eqref{sigmahat} corresponds to local linear smoothing  of $\E(Y_i\mid \t{\V}\X_i)$ with weights given in \eqref{MAVE_weight}.
After the local linear estimates $\hat{a}_j, \hat{\bb}_j$ in \eqref{MAVE_target_func} are obtained,  assume that the weights in \eqref{MAVE_weight} are given by $w_{i,j}(\V) = 1$ if $i =j$ and $0$  if $i \ne j$. Then, the target function of $\mave$ in \eqref{MAVE_target_func} can be written as
\begin{multline}\label{analogy_to_MAVE}
    T_n(\V) = \frac{1}{n} \sum_{j=1}^n \hat{\sigma}^2(\t{\V}\X_j)
            = \frac{1}{n} \sum_i \left(Y_i - \widehat{\E}(Y_i \mid \t{\V} \X_i)\right)^2 \\
            = \frac{1}{n}  \sum_i \Loss\left(Y_i, \widehat{\E}(Y_i \mid \t{\V} \X_i) \right)
\end{multline}
where $\Loss$ is the squared error loss and  $\widehat{\E}(Y_i \mid \t{\V} \X_i)$ the  local linear smooth. Under this simplifying assumption, \eqref{analogy_to_MAVE} is the same as \eqref{NN_target} except that the conditional expectation is estimated via local linear smoothing in $\mave$  as opposed  to neural nets for $\nn$  in \eqref{refinment_training}.

\subsection{Analogy of \texorpdfstring{$\nnopg$}{NNSDR} to \texorpdfstring{$\opg$}{OPG}}
The $\opg$ estimator  estimates \eqref{gradient_matrix} in Theorem~\eqref{thm_OPG} via local linear smoothing of $\nabla g(\t{\B}\X_i)$. 
Specifically, 
if $\V = \I_p$ in \eqref{MAVE_target_func}, we let  $(a_j,\t{\bb_j})_{j=1}^n$ denote the solutions of the optimization  in \eqref{MAVE_target_func}. Then, $\bb_j$ is an estimate for $\nabla g(\t{\B}\X_j)$ and 
\begin{align}\label{OPG_matrix}
    \widehat{\Sigmabf}_\nabla = \frac{1}{n}\sum_{j=1}^n \bb_j\t{\bb_j}
\end{align}
is an estimator for $\Sigmabf_\nabla$ in Theorem~\ref{thm_OPG}.

\begin{definition}
    The \textit{outer product gradient} ($\opg$) estimator for $\ms = \Span\{\B\}$ is defined as
    \begin{align}\label{OPG_estimator}
    \widehat{\B}_{\opg} = (\vb_1,\ldots,\vb_k)
    \end{align}
    where  $\vb_1,\ldots,\vb_k$ are the first $k$ eigenvectors of \eqref{OPG_matrix}.
\end{definition}

An iterative algorithm to solving the optimization problem in \eqref{MAVE_estimator} is given in \cite{Xiaetal2002}. The $\opg$ estimator in \eqref{OPG_estimator} can be used as an initial value for the optimization procedure of the $\mave$ estimator.


\section{Simulations}\label{sec:sim}
We compare the estimation accuracy of $\nn$ estimation  with the forward model based sufficient dimension reduction methods, and \textit{mean outer product gradient estimation} (\texttt{meanOPG}), \textit{mean minimum average variance estimation} (\texttt{meanMAVE} \cite{Xiaetal2002}) and the  recently developed \textit{conditional variance estimator} ($\cve$)  \cite{FertlBura2021a}. The first two, \texttt{meanOPG} and \texttt{meanMAVE},  are implemented in the $\texttt{R}$-package \cite{MAVEpackage}, and $\cve$ in the \texttt{R} package \texttt{CVarE} \cite{CVarE}.  

We report results for three architectures for $\gopg$ and $\gwrap$ used in $\nn$ estimation. The first is a single layer \texttt{MLP} with $128$ hidden neurons, the second has $512$ and the third is a two layer \texttt{MLP} with $48$ hidden neurons each. The results were largely undifferentiated for hidden neuron values between 128 and 512. For the two layer MLP, we obtained similar results for more than 48 neurons, which is already a small number. All three architectures use dropout (see \cite{dropout}) with probability $0.4$\footnote{Dropout rates ranging from 0 to 0.6 were tried and 0.4 was found to yield the best accuracy in reduction estimation.}  after each fully connected hidden layer except in the reduction layer of the $\gnn$. All architectures in \eqref{NN_OPG_MLP} are trained in \eqref{NN_OPG_training} with $ep = 200$ epochs and \texttt{batch\_size} $m = 32$. The refinement training in \eqref{refinment_training} uses $ep = 400$ \texttt{epochs} and again \texttt{batch size} $m = 32$. We use the estimation algorithm in Section~\ref{sec:algo}. The code is available at \url{https://git.art-ist.cc/daniel/NNSDR}.

We consider the same six models (M1-M6) as in \cite{FertlBura2021a}, which are reproduced in Table~\ref{tab:mod}. 
Throughout, we set $p=20$,  $\bb_1 = (1,1,1,1,1,1,0, ...,0)^T/\sqrt{6}$, $\bb_2 = (1,-1,1,-1,1,-1,0,...,0)^T/\sqrt{6} \in \real^p$ for M1-M5. For M6, $\bb_1 = \eb_1, \bb_2 = \eb_2$ and $\bb_3 = \eb_p$
 , where $\eb_j$ denotes the $p$-vector with $j$th  element equal to 1 and all others are 0. In M7, the first three columns are the identity vectors and $\bb_4 = (2\eb_4+\eb_5)/\sqrt{5}$ which is taken from \cite{friedberg_local_2020}. The error term $\epsilon$ is independent of $\X$ for all models. In M2, M3, M4, M5 and M6, $\epsilon \sim N(0,1)$. For M1, $\epsilon$ has a generalized normal distribution $GN(a,b,c)$ with densitiy $f_{\epsilon}(z) = c/(2b\Gamma(1/c))\exp((|z-a|/b)^c)$ [see \cite{gnorm}], with location 0 and shape-parameter 0.5  for M1, and the scale-parameter is chosen such that $\var(\epsilon) = 0.25$. 
The dimension $k$ is assumed to be known throughout. 

\begin{table}
\caption{\label{tab:mod}Models used in the simulations. The distribution of $\epsilon\sim GN(0,\sqrt{1/2},0.5)$ is a generalized normal distribution in Model $M1$. For all others, $\epsilon\sim\mathcal{N}(0, 1)$.}
\centering
\begin{tabular}{lllccc}
\hline\noalign{\smallskip}
Name       & Model                                                               & $\X$ distribution                   & $k$ & $n$ \\
\noalign{\smallskip}\hline\noalign{\smallskip}
  M1${}^a$ & $Y = \cos(\bb_1^T\X) + \epsilon$                                    & $\X \sim N_p(\0,\Sigmabf)$          &  1  & 100 \\
  M2       & $Y = \cos(\bb_1^T\X) + 0.5\epsilon$                                 & $\X \sim Z \1_{p} + N_p(\0,\I_{p})$ &  1  & 100 \\
  M3       & $Y = 2\log(|\bb_1^T\X|+2)+ 0.5\epsilon$                             & $\X \sim N_p(\0,\I_p)$              &  1  & 100 \\
  M4       & $Y = (\bb_1^T\X)/(0.5 +(1.5 + \bb_2^T\X)^2) + 0.5\epsilon$          & $\X \sim N_p(0,\Sigmabf)$           &  2  & 200 \\
  M5       & $Y = \cos(\pi \bb_1^T\X)(\bb_2^T\X + 1)^2 + 0.5\epsilon$            & $\X \sim U([0,1]^p)$                &  2  & 200 \\
  M6       & $Y = (\bb_1^T\X)^2 + (\bb_2^T\X)^2 + (\bb_3^T\X)^2 + 0.5\epsilon$   & $\X \sim N_p(\0,\I_p)$              &  3  & 200 \\
  M7       & $Y = 10\sin(\pi(\t{\bb_1}\X)(\t{\bb_2}\X))$                         & $\X \sim U([0,1]^p)$                &  4  & 600 \\ 
           & $\phantom{Y = }+ 20(\t{\bb_3}\X - 0.5)^2 + 5^{3/2}\t{\bb_4}\X + 5\epsilon$  & & & \\
\noalign{\smallskip}\hline
\end{tabular}
\end{table}

The variance-covariance structure of $\X$ in models M1 and M4 satisfies $\Sigmabf_{i,j} = 0.5^{|i-j|}$ for $i,j=1,\ldots,p$. In M5, $\X$ is uniform  with independent entries on the $p$-dimensional hyper-cube. The link functions of M4 is studied in  \cite{Xiaetal2002}
, but we use $p=20$ instead of 10 and a non identity covariance structure for M4. 
In M2, 
$Z \sim 2\text{Bernoulli}(0.3) - 1 \in \{-1,1\}$, 
where $\1_q = (1,1,...,1)^T\in \real^k$, this yields that $\X$ has a mixture normal distribution with a  mixture probability of $0.3$. M7 is a challenging four dimensional model studied in \cite{friedberg_local_2020}.

We generate $r=100$ replications of models M1 - M7 and estimate $\B$ using the different sufficient dimension reduction methods. The accuracy of the estimates is assessed using  
\begin{align}\label{err}
acc.err&= \frac{\|\Pbf_\B - \Pbf_{\widehat{\B}}\|}{\sqrt{2k}},
\end{align}
which lies in the interval $[0,1]$. The factor $\sqrt{2k}$ normalizes the distance, with values closer to zero indicating better agreement and values closer to one indicating strong disagreement. 

\begin{table}
\caption{\label{tab:summary}Mean and standard deviation of estimation errors for M1-M7}
\centering
\begin{tabular}{l|ccc|ccc}
    \hline\noalign{\smallskip}
          Model   & \opg    & \mave   &  \cve   &$\nn_{128}$&$\nn_{512}$&$\nn_{48, 48}$ \\
    \noalign{\smallskip}\hline\noalign{\smallskip}
M1 &  0.605  &  0.535  &  {\bf 0.396}  &  0.450  &  0.460  &  0.502  \\
   & (0.179) & (0.207) & (0.108) & (0.126) & (0.152) & (0.200) \\ \hline
M2 &  0.918  &  0.910  &  {\bf 0.455}  &  0.635  &  0.619  &  0.752  \\
   & (0.079) & (0.094) & (0.090) & (0.177) & (0.187) & (0.174) \\ \hline
M3 &  0.754  &  0.702  &  0.594  &  0.608  &  {\bf 0.578}  &  0.628  \\
   & (0.216) & (0.258) & (0.209) & (0.211) & (0.196) & (0.228) \\ \hline
M4 &  0.431  &  0.435  &  0.572  &  {\bf 0.408}  &  0.413  &  0.413  \\
   & (0.095) & (0.099) & (0.131) & (0.088) & (0.082) & (0.073) \\ \hline
M5 &  {\bf 0.415}  &  0.422  &  0.441  &  0.547  &  0.554  &  0.601  \\
   & (0.103) & (0.117) & (0.085) & (0.137) & (0.158) & (0.139) \\ \hline
M6 &  0.181  &  0.160  &  0.420  &  0.133  &  {\bf 0.122}  &  0.147  \\
   & (0.027) & (0.022) & (0.111) & (0.015) & (0.013) & (0.017) \\ \hline
M7 &  0.641  &  {\bf 0.637}  &  0.791  &  0.698  &  0.654  &  0.687  \\ 
   & (0.074) & (0.071) & (0.032) & (0.051) & (0.074) & (0.068) \\
    \noalign{\smallskip}\hline
	\end{tabular}
\end{table}

\begin{table}
\caption{\label{tab:summary_MSE} Mean and standard deviation of out of sample-prediction errors for M1-M7 }
\centering
\begin{tabular}{l|ccc|ccc}
    \hline\noalign{\smallskip}
          Model   &  \opg   &  \mave  &  \cve   & $\nn_{128}$ & $\nn_{512}$ & $\nn_{48,48}$ \\
    \noalign{\smallskip}\hline\noalign{\smallskip}
M1 &  0.523  &  0.427       &  {\bf 0.364} &  0.409       &  0.421       &  0.422       \\
   & (0.218) & (0.144)      & (0.059)      & (0.134)      & (0.187)      & (0.172)      \\ \hline
M2 &  0.736  &  0.738       &  {\bf 0.396} &  0.476       &  0.506       &  0.535       \\
   & (0.145) & (0.092)      & (0.044)      & (0.086)      & (0.111)      & (0.110)      \\ \hline
M3 &  0.525  &  0.518       &  0.432       &  0.417       &  0.430       &  {\bf 0.410} \\
   & (0.110) & (0.107)      & (0.083)      & (0.092)      & (0.089)      & (0.088)      \\ \hline
M4 &  0.711  &  0.713       &  0.647       &  {\bf 0.438} &  0.497       &  0.470       \\
   & (0.089) & (0.104)      & (0.096)      & (0.071)      & (0.135)      & (0.062)      \\ \hline
M5 &  0.462  &  0.461       &  {\bf 0.440} &  0.494       &  0.482       &  0.555       \\
   & (0.051) & (0.046)      & (0.043)      & (0.109)      & (0.103)      & (0.099)      \\ \hline
M6 &  0.838  &  0.765       &  2.354       &  0.782       &  {\bf 0.612} &  1.216       \\
   & (0.177) & (0.228)      & (0.914)      & (0.117)      & (0.081)      & (0.224)      \\ \hline
M7 & 33.112  & {\bf 33.066} & 33.884       & 33.955       & 35.272       & 34.136       \\ 
   & (1.961) & (1.973)      & (1.752)      & (1.910)      & (2.383)      & (1.836)      \\
    \noalign{\smallskip}\hline
	\end{tabular}
\end{table}

We report the average $acc.err$ and their standard deviations in Table~\ref{tab:summary}. All three network architectures, $\nn_{128}-\sdr$, $\nn_{512}$, $\nn_{48,48}$ yield similar results, highlighting the robustness of the method with respect to the architecture.  We choose $\nn_{512}$ as our default setup for the following simulations in Section~\ref{sec:bigData}. For M1 and M2, $\cve$ yields the most accurate estimation of the reduction $\B$, followed by the $\nn-\sdr$ estimators. $\opg$ and $\mave$ show the worst performance for the first two models. In M3, the $\nn-\sdr$ estimators are on par with $\cve$ and $\opg$, whereas $\mave$ exhibits the worst performance. In M4, the $\nn-\sdr$ estimators are on par with $\opg$, $\mave$, while $\cve$ is slightly worse than the rest. In M5, $\opg$ and $\mave$ are the most accurate, with $\cve$ nearly on par. For M6, the $\nn-\sdr$ estimators yield the best results followed by $\opg$ and $\mave$. M7 is challenging for all methods, with  $\mave$, $\opg$, and $\nn_{512}-\sdr$ the best performing three.

The $\nn_{512}-\sdr$ estimator is better or on par with $\opg$, $\mave$, and $\cve$ except for M5. 
This is not surprising in the case of $\nn-\sdr$ and $\opg$/$\mave$ as they are built on a similar idea.  The main difference is that $\mave$ uses local linear smoothing instead of neural nets.

Furthermore, in Table~\ref{tab:summary_MSE} we report the mean and standard deviation for the out of sample prediction errors in M1-M7 over $r = 100$ replications. 
For each data set and replication, we sampled a test set with sample size $1000$ from each model and predicted the response $Y$ via the \texttt{predict} function in \texttt{R} for $\opg$, $\mave$, and $\cve$. For $\nn-\sdr$, the predictions are given by $\gnn(\X_{\text{new}}, (\widehat{\B}_{\nn}, \widehat{\Thetabf}_1))$ in \eqref{NN_MLP}. For M1 and M2, $\cve$ gives the smallest out of sample prediction errors, followed by the $\nn-\sdr$ estimators which outperform both $\opg$ and $\mave$. For M3, all three $\nn-\sdr$ estimators are better or on par with $\cve$ and outperform $\opg$ and $\mave$. In M4,  $\nn-\sdr$ outperforms all, with $\cve$ the next best. For M5, $\cve$ performs better than  all other. 
Interestingly, in M6 $\cve$ and $\nn_{48,48}-\sdr$ do not work well in terms of prediction accuracy. In M7, $\mave$ performs the best followed by $\opg$ and $\cve$, but the $\nn-\sdr$ estimators trail closely. 

In sum, for relatively small to medium samples with few predictors $(p=20)$, $\nn-\sdr$ exhibits approximately similar and sometimes better performance than its $\sdr$ competitors.

\section{Large sample size simulation}\label{sec:bigData}
In this section we simulate data from models M6 and M7 and increase both the number of predictors $p$ and the sample size $n$. We monitor the estimation accuracy by $acc.err$ in \eqref{err}
as in Section~\ref{sec:sim}, the out of sample prediction error and the required time for the estimation of a reduction. 

We examined two simulation settings. In the first, we simulated from model M7 using the same $p=20$ and increased the sample size significantly ($n=2^u$, $u=7,9,11,13$).  
The results are displayed in Table~\ref{tab:summary_big_data}. We do not report values for $n=2048, 8192$ for $\cve$ as the runtime is too long. 
For $n \in \{ 128, 512\}$, $\mave$ is on par with $\nn_{512}-\sdr$, whereas for $n=2048, 8192$, $\nn_{512}-\sdr$ is slightly more accurate.

To explore how simultaneous growth of the sample size and the number of predictors affect performance, the second simulation revisits M6, where we successively increase both the sample size $n$ and $p$. The sample sizes considered are $n \in \{1000, 4000, 16000, 64000, 256000\}$ with corresponding $p \in \{32, 63, 126, 253 , 506\}$, which is roughly $p\propto \sqrt{n}$. We observed that for larger sample sizes, fewer epochs in the training phase of the neural net suffice. To demonstrate this, the number of epochs was reduced as $n$ and $p$ increased, as follows. For $(n,p)=(1000,32)$, 200 and 400 epochs were used in the two steps of the refined \nn, respectively, and at each subsequent setting, epoch numbers were halved.

The results of this simulation are shown in Table~\ref{tab:big_data_sim_M6}, which reports the mean and standard deviation (in parentheses) over $10$ repetitions of $acc.err$ in \eqref{err}, the out of sample prediction errors, and the runtime as measured internally via the user time obtained by the \texttt{R} function \texttt{system.time()}. 
The advantage of $\nn$ emerges in Table~\ref{tab:big_data_sim_M6}. As both $n$ and $p$ grow, $\mave$ is no longer computable in realistic time. For example, for $n=64000, p=253$, one calculation for $\mave$ takes about 12 hours to complete. Hence, we report only one value for $acc.err$ and prediction error. In contrast, $\nn$ takes about 9 minutes to complete one run for  the same setting and about 28 minutes to complete one run for $n=256000, p=506$. For $n=1000,4000,16000$, and $p=32,63,126$, $\nn-\sdr$ exhibits slightly higher values of estimation error and lower values of out-of-sample prediction error than $\mave$.

The mean runtimes of the two methods are plotted against the sample size in Figure~\ref{fig:runtime_plot}. We see that the runtime for $\mave$ explodes to exceed $12$ hours only for  one dataset at sample size $64 000$.  On the other hand, $\nn$ computes in reasonable time. 

Thus, $\nn-\sdr$ 
is the only forward model based $\sdr$ method that is applicable to truly large data while obtaining small estimation and out-of-sample prediction errors. Moreover, for smaller data sets, both in terms of $n$ and $p$, it maintains competitive performance.

\begin{table}
    \caption{\label{tab:summary_big_data}Mean and standard deviation (in parentheses) of estimation $acc.err$ for model  M7.}
    \centering
    \begin{tabular}{l|cccc}
    \hline\noalign{\smallskip}
         $n$ &    \opg   &   \mave   &   \cve    &$\nn_{512}$\\
    \noalign{\smallskip}\hline\noalign{\smallskip}
         128 &   0.802   &   \textbf{0.797}   &   0.834   &   0.801   \\
             & (0.02768) & (0.03561) & (0.02567) & (0.03541) \\ \hline
         512 &   0.691   &   \textbf{0.683}   &   0.778   &   0.697   \\
             & (0.05700) & (0.05923) & (0.03528) & (0.03639) \\ \hline
        2048 &   0.233   &   0.253   &           &   \textbf{0.209}   \\
             & (0.03161) & (0.07841) &           & (0.06000) \\ \hline
        8192 &   0.102   &   0.107   &           &   \textbf{0.082}   \\
             & (0.00738) & (0.00935) &           & (0.00722) \\
    \noalign{\smallskip}\hline
    \end{tabular}
\end{table}

\begin{table}
    \caption{\label{tab:big_data_sim_M6}Mean and standard deviation (in parentheses) of $acc.err$, out of sample prediction error, and runtime for model M6.}
    \centering
    \begin{tabular}{ll|cccc}
    \hline\noalign{\smallskip}
       $n$ & $p$ & Method    &  $acc.err$  &     MPE     & time [sec]  \\
    \noalign{\smallskip}\hline\noalign{\smallskip}
      1000 &  32 &\mave      &    0.063    &    0.393    &    5.48     \\
           &     &           &   (0.003)   &   (0.028)   &   (0.031)   \\
           &     &$\nn_{512}$&    0.055    &    0.343    &   48.65     \\
           &     &           &   (0.004)   &   (0.021)   &   (0.700)   \\ \hline
      4000 &  63 &\mave      &    0.045    &    0.351    &   71.20     \\
           &     &           &   (0.002)   &   (0.019)   &   (0.842)   \\
           &     &$\nn_{512}$&    0.050    &    0.313    &   91.35     \\
           &     &           &   (0.003)   &   (0.016)   &   (0.822)   \\ \hline
     16000 & 126 &\mave      &    0.032    &    0.337    & 1416.14     \\
           &     &           &   (0.001)   &   (0.016)   &  (34.367)   \\
           &     &$\nn_{512}$&    0.063    &    0.329    &  215.78     \\
           &     &           &   (0.002)   &   (0.025)   &   (1.793)   \\ \hline
     64000 & 253 &\mave      &    0.023    &    0.325    & $\sim 12 h$ \\
           &     &           &   (0)$^a$   &   (0)$^a$   &   (0)$^a$   \\
           &     &$\nn_{512}$&    0.095    &    0.387    &  542.26     \\
           &     &           &   (0.001)   &   (0.019)   &   (2.934)   \\ \hline
    256000 & 506 &$\nn_{512}$&    0.153    &    0.568    & 1673.03     \\
           &     &           &   (0.003)   &   (0.028)   &   (6.650)   \\
    \noalign{\smallskip}\hline
    \multicolumn{6}{l}{${}^a$ Only one repetition was run as it takes about 12 hours.}
    \end{tabular}
\end{table}

\begin{figure}
    \caption{\label{fig:runtime_plot}Runtime comparison of MAVE against $\nn_{512}$ with equivalent estimation performance.}
    \centering
\begin{tikzpicture}[x=1pt,y=1pt,scale=0.5,every node/.style={scale=0.8}]
\begin{scope}
\path[clip] ( 49.20, 61.20) rectangle (480.69,456.69);

\path[draw=black,line width= 0.4pt,line join=round,line cap=round] ( 65.18,164.72) --
	(165.06,190.36) --
	(264.94,225.34) --
	(364.83,262.84) --
	(464.71,308.68);
\end{scope}
\begin{scope}
\path[clip] (  0.00,  0.00) rectangle (505.89,505.89);

\path[draw=black,line width= 0.4pt,line join=round,line cap=round] ( 65.18, 61.20) -- (446.92, 61.20);
\path[draw=black,line width= 0.4pt,line join=round,line cap=round] ( 65.18, 61.20) -- ( 65.18, 55.20);
\path[draw=black,line width= 0.4pt,line join=round,line cap=round] (115.12, 61.20) -- (115.12, 55.20);
\path[draw=black,line width= 0.4pt,line join=round,line cap=round] (181.14, 61.20) -- (181.14, 55.20);
\path[draw=black,line width= 0.4pt,line join=round,line cap=round] (231.08, 61.20) -- (231.08, 55.20);
\path[draw=black,line width= 0.4pt,line join=round,line cap=round] (281.02, 61.20) -- (281.02, 55.20);
\path[draw=black,line width= 0.4pt,line join=round,line cap=round] (347.04, 61.20) -- (347.04, 55.20);
\path[draw=black,line width= 0.4pt,line join=round,line cap=round] (396.98, 61.20) -- (396.98, 55.20);
\path[draw=black,line width= 0.4pt,line join=round,line cap=round] (446.92, 61.20) -- (446.92, 55.20);
\node[text=black,anchor=base,inner sep=0pt, outer sep=0pt, scale=0.9] at ( 65.18, 39.60) {$1\times 10^{3}$};
\node[text=black,anchor=base,inner sep=0pt, outer sep=0pt, scale=0.9] at (115.12, 39.60) {$2\times 10^{3}$};
\node[text=black,anchor=base,inner sep=0pt, outer sep=0pt, scale=0.9] at (181.14, 39.60) {$5\times 10^{3}$};
\node[text=black,anchor=base,inner sep=0pt, outer sep=0pt, scale=0.9] at (231.08, 39.60) {$1\times 10^{4}$};
\node[text=black,anchor=base,inner sep=0pt, outer sep=0pt, scale=0.9] at (281.02, 39.60) {$2\times 10^{4}$};
\node[text=black,anchor=base,inner sep=0pt, outer sep=0pt, scale=0.9] at (347.04, 39.60) {$5\times 10^{4}$};
\node[text=black,anchor=base,inner sep=0pt, outer sep=0pt, scale=0.9] at (396.98, 39.60) {$1\times 10^{5}$};
\node[text=black,anchor=base,inner sep=0pt, outer sep=0pt, scale=0.9] at (446.92, 39.60) {$2\times 10^{5}$};
\path[draw=black,line width= 0.4pt,line join=round,line cap=round] ( 49.20, 72.13) -- ( 49.20,446.94);
\path[draw=black,line width= 0.4pt,line join=round,line cap=round] ( 49.20, 72.13) -- ( 43.20, 72.13);
\path[draw=black,line width= 0.4pt,line join=round,line cap=round] ( 49.20,100.34) -- ( 43.20,100.34);
\path[draw=black,line width= 0.4pt,line join=round,line cap=round] ( 49.20,165.83) -- ( 43.20,165.83);
\path[draw=black,line width= 0.4pt,line join=round,line cap=round] ( 49.20,194.04) -- ( 43.20,194.04);
\path[draw=black,line width= 0.4pt,line join=round,line cap=round] ( 49.20,259.54) -- ( 43.20,259.54);
\path[draw=black,line width= 0.4pt,line join=round,line cap=round] ( 49.20,287.74) -- ( 43.20,287.74);
\path[draw=black,line width= 0.4pt,line join=round,line cap=round] ( 49.20,353.24) -- ( 43.20,353.24);
\path[draw=black,line width= 0.4pt,line join=round,line cap=round] ( 49.20,381.44) -- ( 43.20,381.44);
\path[draw=black,line width= 0.4pt,line join=round,line cap=round] ( 49.20,446.94) -- ( 43.20,446.94);
\node[text=black,rotate= 90.00,anchor=base,inner sep=0pt, outer sep=0pt, scale=  1.00] at ( 34.80, 72.13) {5};
\node[text=black,rotate= 90.00,anchor=base,inner sep=0pt, outer sep=0pt, scale=  1.00] at ( 34.80,100.34) {10};
\node[text=black,rotate= 90.00,anchor=base,inner sep=0pt, outer sep=0pt, scale=  1.00] at ( 34.80,165.83) {50};
\node[text=black,rotate= 90.00,anchor=base,inner sep=0pt, outer sep=0pt, scale=  1.00] at ( 34.80,194.04) {100};
\node[text=black,rotate= 90.00,anchor=base,inner sep=0pt, outer sep=0pt, scale=  1.00] at ( 34.80,259.54) {500};
\node[text=black,rotate= 90.00,anchor=base,inner sep=0pt, outer sep=0pt, scale=  1.00] at ( 34.80,287.74) {1000};
\node[text=black,rotate= 90.00,anchor=base,inner sep=0pt, outer sep=0pt, scale=  1.00] at ( 34.80,353.24) {5000};
\node[text=black,rotate= 90.00,anchor=base,inner sep=0pt, outer sep=0pt, scale=  1.00] at ( 34.80,446.94) {50000};
\path[draw=black,line width= 0.4pt,line join=round,line cap=round] ( 49.20, 61.20) --
	(480.69, 61.20) --
	(480.69,456.69) --
	( 49.20,456.69) --
	( 49.20, 61.20);
\end{scope}
\begin{scope}
\path[clip] (  0.00,  0.00) rectangle (505.89,505.89);

\node[text=black,anchor=base,inner sep=0pt, outer sep=0pt, scale=  1.20] at (264.94,477.15) {\bfseries loglog Runtime Plot};
\node[text=black,anchor=base,inner sep=0pt, outer sep=0pt, scale=  1.00] at (264.94, 15.60) {n};
\node[text=black,rotate= 90.00,anchor=base,inner sep=0pt, outer sep=0pt, scale=  1.00] at ( 10.80,258.94) {time [sec]};
\end{scope}
\begin{scope}
\path[clip] ( 49.20, 61.20) rectangle (480.69,456.69);

\path[draw=red,line width= 0.4pt,line join=round,line cap=round] ( 65.18, 75.85) -- (165.06,180.22) -- (264.94,301.90) -- (364.83,442.04);



\path[draw=black,line width= 0.4pt,dash pattern=on 4pt off 4pt ,line join=round,line cap=round] ( 49.20,201.46) -- (480.69,201.46);
\path[draw=black,line width= 0.4pt,dash pattern=on 4pt off 4pt ,line join=round,line cap=round] ( 49.20,238.75) -- (480.69,238.75);
\path[draw=black,line width= 0.4pt,dash pattern=on 4pt off 4pt ,line join=round,line cap=round] ( 49.20,339.87) -- (480.69,339.87);
\path[draw=black,line width= 0.4pt,dash pattern=on 4pt off 4pt ,line join=round,line cap=round] ( 49.20,440.99) -- (480.69,440.99);
\node[text=black,anchor=base,inner sep=0pt, outer sep=0pt, scale=1.00] at (115.12,207.46) {$2$ min};
\node[text=black,anchor=base,inner sep=0pt, outer sep=0pt, scale=1.00] at (115.12,244.75) {$5$ min};
\node[text=black,anchor=base,inner sep=0pt, outer sep=0pt, scale=1.00] at (115.12,345.87) {$1$ h};
\node[text=black,anchor=north,inner sep=0pt, outer sep=0pt, scale=1.00] at (115.12,436.0) {$12$ h};
\path[draw=black,line width= 0.4pt,line join=round,line cap=round] (400.34, 97.20) rectangle (480.69, 61.20);

\path[draw=red,line width= 0.4pt,line join=round,line cap=round] (409.34, 85.20) -- (427.34, 85.20);

\path[draw=black,line width= 0.4pt,line join=round,line cap=round] (409.34, 73.20) -- (427.34, 73.20);
\node[text=black,anchor=base west,inner sep=0pt, outer sep=0pt, scale=  1.00] at (436.34, 81.76) {$\mave$};
\node[text=black,anchor=base west,inner sep=0pt, outer sep=0pt, scale=  1.00] at (436.34, 70.51) {$\nn_{512}$};
\end{scope}
\end{tikzpicture}
\end{figure}
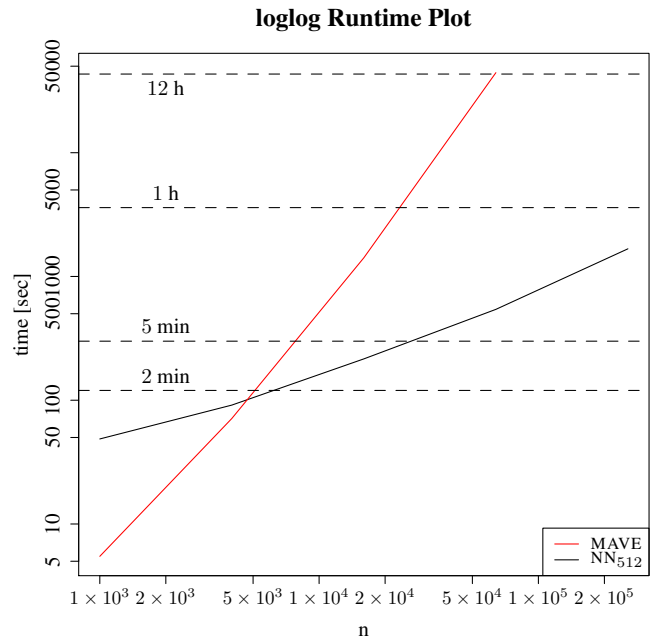

\section{Data Analysis}\label{sec:DataAnalysis}

We analyze three data sets. The first in Section~\ref{sec:BostonHousing} is of relatively small sample size ($n=506$) and number of predictors ($p=12$), the second in Section \ref{sec:KC} is of large $n$ ($21,613$) and small $p$ ($16$), and the third in Section \ref{sec:Beijing} is of very large $n$ ($382,168$) and small to medium $p=40$.

\subsection{Boston Housing}\label{sec:BostonHousing}
In this section we apply 
the refined $\nn$ estimator on the \texttt{Boston Housing} data and compare its performance with the other two mean subspace SDR methods, $\mave$ and $\cve$. 
This data set has been extensively used as a benchmark for assessing regression methods [see, for example, \cite{James2013}], and is available in the \texttt{R}-package \texttt{mlbench}. 
The data comprise of 506 instances of 14 variables from the  1970 Boston census, 13 of which are continuous. The binary variable \texttt{chas}, indexing proximity to the Charles river, is omitted from the analysis since all three methods operate under the assumption of continuous predictors. The target variable is the median value of owner-occupied homes, \texttt{medv}, in $\$1,000$. The 12 predictors are \texttt{crim} (per capita crime rate by town), \texttt{zn} (proportion of residential land zoned for lots over 25,000 sq.ft), \texttt{indus} 	(proportion of non-retail business acres per town),
\texttt{nox} (nitric oxides concentration (parts per 10 million)), \texttt{rm} (average number of rooms per dwelling), \texttt{age} (proportion of owner-occupied units built prior to 1940), \texttt{dis} (weighted distances to five Boston employment centres), \texttt{rad} (index of accessibility to radial highways), \texttt{tax} (full-value property-tax rate per $\$10,000$), \texttt{ptratio} (pupil-teacher ratio by town),  \texttt{lstat} (percentage of lower status of the population), and \texttt{b} stands for $1000(B - 0.63)^2$ where $B$ is the proportion of blacks by town.

We set the dimension of the reduction $\B$ to two; i.e., $k = 2$, for all three methods and  compute prediction errors using squared error loss and leave-one-out cross validation. The $\nn$ with one layer and 512 neurons is fitted on the $n-1$  training data and compute the predicted value for the left out data point. Both $\cve$ and $\mave$ were applied to the standardized  training data. 
The mean and standard deviation (in parentheses) of the 506 prediction errors are displayed in Table~\ref{tab:Boston}.  The $\cve$ method results in the smallest prediction error followed by $\nn-\sdr$, which, on the other hand, has the smallest standard error. $\mave$ is the least accurate.  The analysis for $k = 1$ yielded similar results.
In this example of small $n$-small $p$, nonparametric methods are expected to do well, which is what we observe for $\cve$ followed by $\mave$. Nevertheless, the performance of the large sample $\nn-\sdr$ method is roughly on par with both. 

\begin{table}
    \caption{\label{tab:Boston}Leave-One-Out Cross Validation Prediction errors mean and standard deviation (in brackets) with reduction dimension $k = 2$.}
    \centering
    \begin{tabular}{r | ccc}
        \hline\noalign{\smallskip}
               &  $\mave$ &  $\cve$  &$\nn_{512}$ \\
        \noalign{\smallskip}\hline\noalign{\smallskip}
          mean &  18.762  &  16.148  &  18.006  \\ 
          (sd) & (63.136) & (63.500) & (41.739) \\
        \noalign{\smallskip}\hline
    \end{tabular}
\end{table}

\subsection{KC Housing}\label{sec:KC}
Further, we use \texttt{kc\_house\_data} set in the \texttt{R} package $\mave$ to compare $\nn-\sdr$ estimation with  $\mave$. 
The data set contains $21613$ observations on $20$ variables. The target variable is \texttt{price}, the price of a sold house.
We use 16 predictors after omitting \texttt{id, date}, and \texttt{zip code}: 
\texttt{bedrooms} (number of bedrooms),  \texttt{bathrooms} (number of bathrooms), \texttt{sqft\_living} (square footage of the living room(, \texttt{sqrt\_log} (square footage of the log),  \texttt{floors} (total floors in the house), \texttt{waterfront} (whether the house has a view a waterfront(1: yes, 0: not)), \texttt{view} (unknown), \texttt{condtion} (condition of the house), \texttt{grade} (unknown),  \texttt{sqft\_above} (square footage of house apart from basement), \texttt{sqft\_basement} (square footage of the basement),   \texttt{yr\_built} (built year), \texttt{yr\_renovated} (year when the house was renovated),   
\texttt{lat} (latitude coordinate), \texttt{long} (longitude coordinate), \texttt{sqft\_living15} (living room area in 2015(implies some renovations)), \texttt{sqrt\_lot15} (lot area in 2015(implies some renovations)).

We perform 10-fold cross-validation in order to obtain an unbiased estimate of the out of sample prediction error. We set $k = 1$ and report the average fraction of the mean squared prediction error divided by the variance of the response on the test set, as well as its standard error, in Table~\ref{tab:KC}. Our $\nn-\sdr$ estimator has  out of sample mean squared error that is about half the variance of the response on the test set, whereas $\mave$'s is less than 2 percent lower than the variance of the response. This means that the $\nn-\sdr$ regression explains roughly half of the total variance in the response whereas $\mave$ hardly explains any. Further, even though the $\mave$ reduction is estimated in roughly the same time as $\nn-\sdr$, in $6$ out of the $10$ folds the \texttt{predict} function for $\mave$ produces an error. We also report the 10-fold cross-validated prediction error for $\cve$, which yields  the best result as it explains more than  70\% of the total variance in the response but could not be computed, in its current implementation, on a personal computer. \footnote{The $\cve$ values were computed on the \emph{Vienna Scientific Cluster} (VSC).}

The coefficients of the reductions are given in Table \ref{tab:coefs}. $\nn-\sdr$ extracts information from all variables as it places non-zero weights of varying size on all. $\mave$, on the other hand, selects \texttt{waterfront} and the co-linear \texttt{sqft\_living, sqft\_above, sqft\_basement} (\texttt{sqft\_living = sqft\_above + sqft\_basement}) and drops all other variables. Moreover, it allocates the same weight to the collinear variables with opposite signs, effectively discounting all three and ultimately declaring only \texttt{waterfront} relevant.

These results indicate that $\mave$ breaks down in the analysis of this data set. Since \texttt{sqft\_living} = \texttt{sqft\_above} + \texttt{sqft\_basement}, we dropped \texttt{sqft\_basement} to investigate the effect of collinearity. In Figure \ref{fig:kc_house_data_reduced}, we plot the response versus the $\mave$ reduction computed on all predictors in the left panel, versus the $\mave$ reduction without \texttt{sqft\_basement} and versus the $\nn-\sdr$ reduction in the right panel. The reduced predictors are strikingly different. The $\nn-\sdr$  reduction is smooth and captures a clear nonlinear heteroskedastic relationship with \texttt{price}.  
The plot in the left panel captures the failure of $\mave$ to extract the predictive information in the predictors, as no apparent pattern emerges. Moreover, the data are arbitrarily split in the groups defined by the binary \texttt{waterfront} variable. Once the collinearity is removed, $\mave$ captures the relationship between $Y$ and $\X$ but nevertheless it again splits the data into two new arbitrary classes for the renovated and non-renovated houses. This variable takes either value 0 (not renovated) or the renovation year that ranges between 1934 and 2015. The black points in the middle panel correspond to 0 and red to the period 1934-2015.  

We further draw attention to the semblance of the data clouds across the two categories in the middle panel and the $\nn-\sdr$ reduction in the right panel. Both $\mave$ and $\nn-\sdr$ discover the same pattern, with the correlation coefficients of $\mave$ and $\nn-\sdr$ reductions being 0.82 and 0.85, albeit $\mave$ introduces an artificial split in the data. 

In Table \ref{tab:coefs}, we also provide the coefficients of the last two eigenvectors, corresponding to the two smallest eigenvalues in decreasing order, of the sample covariance matrix of the predictors. The next to last places most of the weight on \texttt{waterfront} and the last on \texttt{sqft\_living} and \texttt{sqft\_above, sqft\_basement}. Moreover, the vector of coefficients of the $\mave$ reduction based on all predictors in the first column seems to be the sum of the last and the down-weighted second to last eigenvectors of the sample covariance matrix of $\X$. This relates to the fact that the sample covariance matrix of $\X$ is singular of rank $16=p-1$. Thus, the last eigenvector dominates all others and largely agrees with the $\mave$ reduction coefficients. We investigate the effect of collinearity on $\mave$ and $\cve$ in Section \ref{sec:singular}. 

\begin{table}
    \caption{\label{tab:KC} Ten-fold Cross Validation Relative Prediction errors mean and standard deviation (in brackets) with reduction dimension $k = 1$.}
    \centering
    \begin{tabular}{r| ccc}
        \hline\noalign{\smallskip}
      &        $\mave$        &         $\cve$        &      $\nn_{512}$      \\
      \noalign{\smallskip}\hline\noalign{\smallskip}
  mean &  0.982  &  0.296  &  0.527  \\
  (sd) & (0.035) & (0.149) & (0.043) \\
\noalign{\smallskip}\hline
    \end{tabular}
\end{table}

\begin{figure}
    \centering
    \caption{\label{fig:kc_house_data_reduced}Reduced data versus response of the \texttt{kc\_house\_data} for $\mave$ and $\nn_{512}$.}
    \includegraphics[width=\columnwidth]{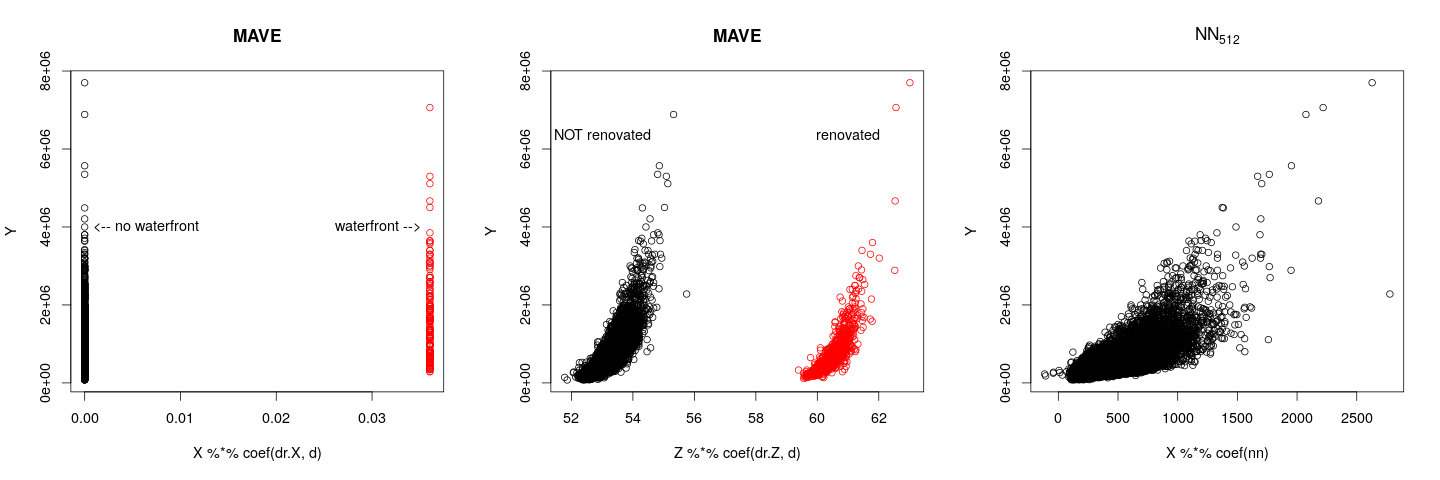}
\end{figure}

\begin{table}[!ht]
    \caption{\label{tab:coefs}Estimated linear reductions for the \texttt{kc\_house\_data} data set. $\Z$ are the remaining $16$ predictors when dropping \texttt{sqft\_basement} from $\X$. The last column $v_{17}(\X), v_{16}(\X)$ are the PCA coefficients corresponding to the last two smallest eigenvalue of $\var(\X)$ for un-scaled $\X$.}
    \centering
    \begin{tabular}{l| ccccc}
        \hline\noalign{\smallskip}
              &$\hat{\B}_{\mave(\X)}$&$\hat{\B}_{\mave(\Z)}$&$\hat{\B}_{\nn_{512}(\X)}$&$v_{17}(\X)$&$v_{16}(\X)$ \\
        \noalign{\smallskip}\hline\noalign{\smallskip}
bedrooms      &  0.000 & -0.015 & -0.171 &  0.000 & -0.005  \\
bathrooms     &  0.000 &  0.050 &  0.057 &  0.000 & -0.001  \\
sqft\_living  &  0.577 &  0.000 &  0.099 &  0.577 &  0.000  \\
sqft\_lot     &  0.000 &  0.000 &  0.000 &  0.000 &  0.000  \\
floors        &  0.000 &  0.036 &  0.098 &  0.000 &  0.000  \\
waterfront    &  0.036 &  0.330 &  0.485 &  0.000 & -0.999  \\
view          &  0.000 &  0.047 &  0.787 &  0.000 &  0.046  \\
condition     &  0.000 &  0.042 &  0.152 &  0.000 &  0.002  \\
grade         &  0.000 &  0.124 &  0.204 &  0.000 & -0.003  \\
sqft\_above   & -0.577 &  0.000 &  0.080 & -0.577 &  0.000  \\
sqft\_basement& -0.577 &        &  0.112 & -0.577 &  0.000  \\
yr\_built     &  0.000 & -0.003 & -0.031 &  0.000 &  0.000  \\
yr\_renovated &  0.000 &  0.004 &  0.051 &  0.000 &  0.000  \\
lat           &  0.000 &  0.925 & -0.001 &  0.000 & -0.011  \\
long          &  0.000 & -0.107 &  0.007 &  0.000 & -0.016  \\
sqft\_living15&  0.000 &  0.000 &  0.079 &  0.000 &  0.000  \\
sqft\_lot15   &  0.000 &  0.000 & -0.001 &  0.000 &  0.000  \\
        \noalign{\smallskip}\hline
    \end{tabular}
\end{table}

\subsubsection{The case of singular \texorpdfstring{$\Sigmaxbf$}{Sigma}}\label{sec:singular}
We consider the effect of collinear predictors on the sufficient dimension reduction
techniques $\mave$, $\cve$, and $\nn-\sdr$. We assume that $\Sigmaxbf = \var(\X)$ is singular
and show that, in this case, the mean subspace is not uniquely identifiable.

Let $\Ub$ be a basis of the nullspace of $\Sigmaxbf$, consisting of the eigenvectors
that correspond to the 0 eigenvalue. Without loss of generality, we assume the
eigenspace of $\Ub$ to be one dimensional. Then $\t{\Ub}\X = c$ is constant and we can write 
\begin{multline}\label{Bnotunique}
    Y = g(\t{\B}\X + c - c) + \epsilon
      = g_c(\t{(\B+ \Ub)}\X) + \epsilon \\
      = g_c(\t{\tilde{\B}}\X) + \epsilon
\end{multline}
where $g_c(\xn) = g(\xn - c)$ fulfills all assumptions of the link function  in model~\eqref{mod:mean_basic} and $\tilde{\B} = \B + \Ub$. 
If $\Span\{\Ub\} \subset \Span\{\B\}$, then $\Span\{\B\} = \Span\{\tilde{\B}\}$ and the mean subspace is unique. Otherwise, both $\Span\{\B\}$ and $\Span\{\tilde{\B}\}$ are dimension reduction subspaces but $\Span\{\tilde{\B}\} \neq \Span\{\B\}$.

Most $\sdr$ approaches, including $\mave$ \cite[Cond. 3(a), p. 386]{Xiaetal2002} and $\cve$ \cite[Cond. A.1 , p. 3 ]{FertlBura2021a} require $\X$ have a density; that is, its variance-covariance is positive definite. It appears that $\mave$ is more sensitive to the violation of this assumption as compared to $\cve$.


To demonstrate this we present a small simulation study. Let $\X = \t{(X_1, \ldots,X_p)}$ with $(X_2,\ldots,X_p) \sim \normal(\0,\I_{p-1})$ and $X_1 = -0.5(X_2 +X_3) + 0.001 Z$, where $Z \sim \normal(0,1)$ and is independent of $(X_2,\ldots,X_p)$. Then, $\Ub = (2,1,1,0,\ldots,0)/\sqrt{6} \approx (0.816, 0.408, 0.408, 0,\ldots,0)$ is the eigenvector of $\Sigmaxbf$ corresponding to the smallest eigenvalue.

Let $p = 10$ and $Y = (\t{\B}\X)^2 + 0.5\epsilon $, where $\epsilon \sim  \normal (0,1)$ is independent from $\X$ and $\B = \eb_4$ the fourth standard basis vector. We draw $100$ random samples $\t{(Y_i, \t{\X}_i)}_{i=1,...,n}$ of size $n = 100$ from this model and calculate the $\mave$, $\cve$ and $\nn_{512}$ estimators of $\B$. The median, mean and  standard deviation of the estimation errors for the subspace in \eqref{err} are reported in Table~\ref{tab:singular}.

\begin{table}[!ht]
    \centering
    \caption{\label{tab:singular}Ten-fold Cross Validation Relative Prediction errors mean and standard deviation (in brackets) with strong collinearity in the predictors.}
    \begin{tabular}{r| ccc}
        \hline\noalign{\smallskip}
             & $\mave$ & $\cve$ & $\nn_{512}$  \\
        \noalign{\smallskip}\hline\noalign{\smallskip}
            mean   & 0.917 & 0.164 &  0.101 \\
            median & 0.999 & 0.162 &  0.096 \\
            (sd)   &(0.256)&(0.057)& (0.032)\\
        \noalign{\smallskip}\hline
    \end{tabular}
\end{table}

For example, one  of the $\B_{\mave}$ estimates is $(0.815$, $0.404$, $0.402$, $-0.102$, $0$,$0$, $-0.001$, $-0.006$, $0.003$, $-0.008)$ with associated error $0.995$. We  can clearly see that $\mave$ estimates $\Ub$ instead of $\B$, and  most $\mave$ estimates follow the same pattern. On the other hand, one  of the $\B_{\cve}$ estimates is $(0.013$, $-0.007$, $-0.018$, $-0.99$, $0.015$, $-0.004$, $-0.046$, $-0.13$, $0.015$, $-0.021)$, with associated error $0.143$ and one  of the $\B_{\nn_{512}}$ estimates is $(-0.007$, $-0.002$, $-0.074$, $-0.995$, $0.013$, $0.004$, $-0.035$, $-0.047$, $0.026$, $-0.031)$, with associated error $0.103$. $\cve$ and $\nn-\sdr$ stays clear of $\Ub$ and correctly identifies the true $\B$.

In this example, in particular, $\mave$ seems to focus solely on estimating $\Ub$ instead of $\B$. This does not hold in general. We offer an explanation by setting $c = \alpha c$ in \eqref{Bnotunique} for a scalar $\alpha$. Following the rationale below \eqref{Bnotunique}, $\tilde{\B} = \B + \alpha \Ub$ is a reduction for any $\alpha$. Since $\mave$, $\cve$ and $\nn-\sdr$ work with $\tilde{\B} \in \spc(p,k)$, 
$\alpha$ determines the weight placed on $\Ub$ relative to $\B$. For large $\alpha$, $\Ub$ dominates the reduction $\tilde{\B}$ and $\mave$ fails to identify the mean subspace.
 In contrast, $\cve$ and $\nn-\sdr$ remains robust in its ability to accurately estimate the mean subspace. 

We conjecture that $\mave$'s vulnerability is numerical in nature and relates to the implementation algorithm in the $\mave$ package.
We also conjecture that $\cve$ and $\nn-\sdr$ are more robust than $\mave$. 

\subsection{Beijing Air Quality Data}\label{sec:Beijing}
The \emph{Beijing Multi-Site Air-Quality Data} \cite{zhang_beijing_air_2017} available at the UCI machine learning repository\footnote{\url{https://archive.ics.uci.edu/ml/datasets/Beijing+Multi-Site+Air-Quality+Data}} includes hourly air pollutants data from 12 nationally-controlled air-quality monitoring sites in Beijing. The air-quality data are from the Beijing Municipal Environmental Monitoring Center. The meteorological data in each air-quality site are matched with the nearest weather station from the China Meteorological Administration. After removing missing data entries, the data contains $382168$ complete measurements.

The target $PM_{2.5} [ug/m^3]$ is the concentration of particle matter in the air with less than $2.5$ micrometres in diameter.

The predictors are \texttt{year}, \texttt{month}, \texttt{day}, \texttt{hour}, \texttt{SO2} (SO2 concentration), \texttt{NO2} (NO2 concentration), \texttt{CO} (CO concentration), \texttt{O3} (O3 concentration), \texttt{TEMP} (temperature), \texttt{PRES} (pressure), \texttt{DEWP} (dew point temperature), \texttt{RAIN} (precipitation), \texttt{wd} (wind direction), \texttt{WSPM} (wind speed), \texttt{station} (name of the air-quality monitoring site). The two categorical variables \texttt{wd} and \texttt{station}, with 16 and 12 categories, respectively, are converted to 26 dummy variables, 
resulting in $40$ predictors.

We included the categorical variables to demonstrate that $\nn-\sdr$ can handle dummy variables even though it is not designed for this. Given the large sample size we used $2$ epochs for the first stage and $3$ for the second refinement stage of the training. Due to the large sample size $\mave$ and $\cve$ are infeasible to compute while $\nn$-SDR executes in less than 3 minutes per fold run on the CPU of personal computer.
As a comparison, we included the \emph{linear model} (\texttt{lm}) as well as the \emph{Multivariate Adaptive Regression Splines} (\texttt{mars}, \cite{mars,mdaRpackage}), as both can be applied to large regressions, provided $p < n$ and are computationally efficient. 

In Table~\ref{tab:Beijing_Air}, the mean of the 10-fold cross validation prediction errors is reported. 
The linear model exhibits the worst performance, as expected. $\nn-\sdr$ improves upon the linear model for all choices of dimension we examined, beats \texttt{mars} for $k=3,4$ and obtains the minimum MSPE for $k=4$. Thus, not only is $\nn-\sdr$ the best method with respect to predictive accuracy, but it also provides an assessment of the true structural dimension of the relationship ($k=4)$ between the response and the predictors. This confirms the improved performance of \texttt{mars}, a multivariate nonparametric fitting method, over the linear model and points to the nonlinearity of the relationship. 

\begin{table}[!ht]
    \caption{\label{tab:Beijing_Air}10-Fold Cross Validation Mean Squared Prediction errors.}
    \centering
    \begin{tabular}{r | cccccc}
        \hline\noalign{\smallskip}
      &  \texttt{lm}  &  \texttt{mars}  &$\nn_{512}$&$\nn_{512}$&$\nn_{512}$&$\nn_{512}$ \\
      &         &        &$k = 1$ &$k = 2$ &$k = 3$ &$k = 4$\\
      \noalign{\smallskip}\hline\noalign{\smallskip}
 mean &   1829  &  1628  &  1746  &  1654  &  1604  &  1526  \\
 (sd) &  (20.8) & (24.9) & (19.9) & (18.8) & (24.0) & (59.3) \\
        \noalign{\smallskip}\hline
    \end{tabular}
\end{table}

\section{Discussion}\label{sec:discussion}
We introduced the novel $\nn-\sdr$ estimator for the mean subspace.
$\nn-\sdr$ combines a sufficient dimension reduction approach with neural nets to
first reduce the predictor vector and then estimate its functional relationship
with the response and predict it. The estimator is shown to be competitive with
state-of-the-art $\sdr$ approaches, such as $\mave$ and $\cve$, in simulations
and data applications. Moreover, it is the only one among them that is computationally
feasible for big data, where both $p$ and $n$ are large, even on personal computers. 

In view of our simulation results, the $\nn-\sdr$ estimator appears to be consistent.
Nevertheless, we could not resolve the theoretical challenges involving neural
nets to formally prove consistency, as this would require showing the consistency
of neural net estimates, which  remains an open problem. 

A particularly attractive feature of $\nn-\sdr$, in contrast to $\mave$ and $\cve$,
is that it is naturally configured for online training if new data become available
due to the stochastic gradient descent algorithm described in Section~\ref{sec:algo}.
Specifically, the algorithm effortlessly updates the parameters of $\nn-\sdr$
with further gradient steps using the new data.

\section*{Acknowledgements}
The authors gratefully acknowledge the support of the Austrian Science Fund (FWF P 30690-N35).
The computations for $\cve$ in Table \ref{tab:KC} were carried out using the Vienna Scientific Cluster (VSC).


\begin{thebibliography}{10}
\bibitem{tensorflow2015-whitepaper}
Mart\'{\i}n Abadi, Ashish Agarwal, Paul Barham, Eugene Brevdo, Zhifeng Chen,
    Craig Citro, Greg~S. Corrado, Andy Davis, Jeffrey Dean, Matthieu Devin,
    Sanjay Ghemawat, Ian Goodfellow, Andrew Harp, Geoffrey Irving, Michael Isard,
    Yangqing Jia, Rafal Jozefowicz, Lukasz Kaiser, Manjunath Kudlur, Josh
    Levenberg, Dandelion Man\'{e}, Rajat Monga, Sherry Moore, Derek Murray, Chris
    Olah, Mike Schuster, Jonathon Shlens, Benoit Steiner, Ilya Sutskever, Kunal
    Talwar, Paul Tucker, Vincent Vanhoucke, Vijay Vasudevan, Fernanda Vi\'{e}gas,
    Oriol Vinyals, Pete Warden, Martin Wattenberg, Martin Wicke, Yuan Yu, and
    Xiaoqiang Zheng.
\newblock {TensorFlow}: Large-scale machine learning on heterogeneous systems,
    2015.
\newblock Software available from tensorflow.org.

\bibitem{tensorflow}
JJ~Allaire and Yuan Tang.
\newblock {\em tensorflow: R Interface to 'TensorFlow'}, 2020.
\newblock {R} package version 2.2.0.

\bibitem{Bottou1998SGD}
L\'{e}on Bottou.
\newblock Online algorithms and stochastic approximations.
\newblock In David Saad, editor, {\em Online Learning and Neural Networks}.
    Cambridge University Press, Cambridge, UK, 1998.
\newblock \url{http://leon.bottou.org/papers/bottou-98x} revised, Oct 2012.

\bibitem{BuraCook2001}
Efstathia Bura and R.~Dennis Cook.
\newblock Estimating the structural dimension of regressions via parametric
    inverse regression.
\newblock {\em Journal of the Royal Statistical Society. Series B: Statistical
    Methodology}, 63(2):393--410, 2001.

\bibitem{BuraDuarteForzani2016}
Efstathia Bura, Sabrina Duarte, and Liliana Forzani.
\newblock Sufficient reductions in regressions with exponential family inverse
    predictors.
\newblock {\em Journal of the American Statistical Association},
    111(515):1313--1329, 2016.
\newblock \url{https://doi.org/10.1080/01621459.2015.1093944}.

\bibitem{BuraForzani2015}
Efstathia Bura and Liliana Forzani.
\newblock Sufficient reductions in regressions with elliptically contoured
    inverse predictors.
\newblock {\em Journal of the American Statistical Association},
    110(509):420--434, 2015.
\newblock \url{https://doi.org/10.1080/01621459.2014.914440}.

\bibitem{Cook1998b}
R.~Dennis Cook.
\newblock Principal hessian directions revisited.
\newblock {\em Journal of the American Statistical Association},
    93(441):84--94, 1998.

\bibitem{Cook2000}
R.~Dennis Cook.
\newblock Save: A method for dimension reduction and graphics in regression.
\newblock {\em Communications in Statistics - Theory and Methods},
    29:2109--2121, 09 2000.
\newblock \url{https://doi.org/10.1080/03610920008832598}.

\bibitem{Cook2007}
R.~Dennis Cook.
\newblock {Fisher Lecture: Dimension Reduction in Regression}.
\newblock {\em Statistical Science}, 22(1):1--26, 02 2007.
\newblock \url{https://doi.org/10.1214/088342306000000682}.

\bibitem{CookForzani2008}
R.~Dennis Cook and L.~Forzani.
\newblock Principal fitted components for dimension reduction in regression.
\newblock {\em Statistical Science}, 23(4):485--501, 2008.

\bibitem{CookForzani2009}
R.~Dennis Cook and Liliana Forzani.
\newblock Likelihood-based sufficient dimension reduction.
\newblock {\em Journal of the American Statistical Association},
    104(485):197--208, 3 2009.
\newblock \url{https://doi.org/10.1198/jasa.2009.0106}.

\bibitem{CookLi2002}
R.~Dennis Cook and Bing Li.
\newblock {Dimension reduction for conditional mean in regression}.
\newblock {\em The Annals of Statistics}, 30(2):455 -- 474, 2002.
\newblock \url{https://doi.org/10.1214/aos/1021379861}.

\bibitem{CookLi2004}
R.~Dennis Cook and Bing Li.
\newblock {Determining the dimension of iterative Hessian transformation}.
\newblock {\em The Annals of Statistics}, 32(6):2501 -- 2531, 2004.
\newblock \url{https://doi.org/10.1214/009053604000000661}.

\bibitem{FertlBura2021a}
Lukas Fertl and Efstathia Bura.
\newblock Conditional {Variance} {Estimator} for {Sufficient} {Dimension}
    {Reduction}.
\newblock arXiv:2102.08782 [math, stat], 2021.
\newblock \url{http://arxiv.org/abs/2102.08782}.

\bibitem{FertlBura2021b}
Lukas Fertl and Efstathia Bura.
\newblock Ensemble {Conditional} {Variance} {Estimator} for {Sufficient}
    {Dimension} {Reduction}.
\newblock arXiv:2102.13435 [stat], 2021.
\newblock \url{https://arxiv.org/abs/2102.13435}.

\bibitem{friedberg_local_2020}
Rina Friedberg, Julie Tibshirani, Susan Athey, and Stefan Wager.
\newblock Local {L}inear {F}orests.
\newblock arXiv:1807.11408 [cs, econ, math, stat], 2020.
\newblock \url{http://arxiv.org/abs/1807.11408}.

\bibitem{mars}
Jerome~H. Friedman.
\newblock Multivariate adaptive regression splines.
\newblock {\em The Annals of Statistics}, 19(1):1--67, 1991.

\bibitem{Goodfellow-et-al-2016}
Ian Goodfellow, Yoshua Bengio, and Aaron Courville.
\newblock {\em Deep Learning}.
\newblock MIT Press, 2016.
\newblock \url{http://www.deeplearningbook.org}.

\bibitem{NN1}
Kevin Gurney.
\newblock {\em An Introduction to Neural Networks}.
\newblock Taylor \& Francis, Inc., USA, 1997.

\bibitem{mdaRpackage}
Trevor Hastie and Tibshirani Robert.
\newblock {\em mda: Mixture and Flexible Discriminant Analysis}, 2017.
\newblock {S} original by {Trevor} {Hastie} \& {Robert} {Tibshirani}.
    {Original} {R} port by {Friedrich} {Leisch} and {Kurt} {Hornik} and {Brian D.
    Ripley}. {R} package version 0.4-10.

\bibitem{HORNIK1991251}
Kurt Hornik.
\newblock Approximation capabilities of multilayer feedforward networks.
\newblock {\em Neural Networks}, 4(2):251--257, 1991.
\newblock \url{https://doi.org/10.1016/0893-6080(91)90009-T}.

\bibitem{James2013}
Gareth James, Daniela Witten, Trevor Hastie, and Robert Tibshirani.
\newblock {\em An {Introduction} to {Statistical} {Learning}: with
    {Applications} in {R}}.
\newblock Number 103 in Springer texts in statistics. Springer, New York, 2013.

\bibitem{ISL}
Gareth James, Daniela Witten, Trevor Hastie, and Robert Tibshirani.
\newblock {\em An Introduction to Statistical Learning: With Applications in
    R}.
\newblock Springer Publishing Company, Incorporated, 2014.

\bibitem{CVarE}
Daniel Kapla and Lukas Fertl.
\newblock {\em {CVarE}: Conditional Variance Estimator for Sufficient Dimension
    Reduction}, 2021.
\newblock {R} package version 1.1.

\bibitem{NNAutoencoder}
Diederik~P. Kingma and Max Welling.
\newblock An {I}ntroduction to {V}ariational {A}utoencoders.
\newblock arXiv:1906.02691 [cs.LG], 2019.
\newblock \url{http://arxiv.org/abs/1906.02691}.

\bibitem{NNAutoencoder2}
Mark~A. Kramer.
\newblock Nonlinear principal component analysis using autoassociative neural
    networks.
\newblock {\em AIChE Journal}, 37(2):233--243, 1991.
\newblock \url{https://doi.org/10.1002/aic.690370209}.

\bibitem{Li2018}
Bing Li.
\newblock {\em Sufficient dimension reduction: methods and applications with
    R}.
\newblock CRC Press, Taylor \& Francis Group, 2018.

\bibitem{directionalRegression}
Bing Li and Shaoli Wang.
\newblock On directional regression for dimension reduction.
\newblock {\em Journal of the American Statistical Association},
    102(479):997--1008, 2007.
\newblock \url{https://doi.org/10.1198/016214507000000536}.

\bibitem{contourRegression}
Bing Li, Hongyuan Zha, and Francesca Chiaromonte.
\newblock {Contour regression: A general approach to dimension reduction}.
\newblock {\em The Annals of Statistics}, 33(4):1580--1616, 2005.
\newblock \url{https://doi.org/10.1214/009053605000000192}.

\bibitem{Li1991}
Ker-Chau Li.
\newblock Sliced inverse regression for dimension reduction.
\newblock {\em Journal of the American Statistical Association},
    86(414):316--327, 1991.

\bibitem{Li1992}
Ker-Chau Li.
\newblock On principal hessian directions for data visualization and dimension
    reduction: Another application of stein's lemma.
\newblock {\em Journal of the American Statistical Association},
    87(420):1025--1039, 1992.
\newblock \url{https://doi.org/10.1080/01621459.1992.10476258}.

\bibitem{MaZhu2013}
Yanyuan Ma and Liping Zhu.
\newblock A {R}eview on {D}imension {R}eduction.
\newblock {\em International Statistical Review}, 81(1):134--150, 4 2013.
\newblock \url{https://doi.org/10.1111/j.1751-5823.2012.00182.x}.

\bibitem{NN2}
Warren~S McCulloch and Walter Pitts.
\newblock A logical calculus of the ideas immanent in nervous activity.
\newblock {\em The bulletin of mathematical biophysics}, 5(4):115--133, 1943.

\bibitem{gnorm}
Saralees Nadarajah.
\newblock A generalized normal distribution.
\newblock {\em Journal of Applied Statistics}, 32(7):685--694, 2005.
\newblock \url{https://doi.org/10.1080/02664760500079464}.

\bibitem{dropout}
Nitish Srivastava, Geoffrey Hinton, Alex Krizhevsky, Ilya Sutskever, and Ruslan
    Salakhutdinov.
\newblock Dropout: A simple way to prevent neural networks from overfitting.
\newblock {\em Journal of Machine Learning Research}, 15(56):1929--1958, 2014.

\bibitem{WangXia2008}
Hansheng Wang and Yingcun Xia.
\newblock Sliced regression for dimension reduction.
\newblock {\em Journal of the American Statistical Association},
    103(482):811--821, 2008.
\newblock \url{https://doi.org/10.1198/016214508000000418}.

\bibitem{MAVEpackage}
Hang Weiqiang and Xia Yingcun.
\newblock {\em {MAVE}: {Methods} for {Dimension} {Reduction}}, 2019.
\newblock {R} package version 1.3.10.

\bibitem{hinton2012RMSprop}
Geoffrey~Hinton with Nitish~Srivastava and Kevin Swersky.
\newblock {Neural Networks for Machine Learning Lecture 6a - Overview of
    mini-batch gradient descent}, 2012.

\bibitem{Xia2007}
Yingcun Xia.
\newblock A constructive approach to the estimation of dimension reduction
    directions.
\newblock {\em Ann. Statist.}, 35(6):2654--2690, 12 2007.
\newblock \url{https://doi.org/10.1214/009053607000000352}.

\bibitem{Xiaetal2002}
Yingcun Xia, Howell Tong, W.~K. Li, and Li-Xing Zhu.
\newblock An adaptive estimation of dimension reduction space.
\newblock {\em Journal of the Royal Statistical Society: Series B (Statistical
    Methodology)}, 64(3):363--410, 2002.
\newblock \url{https://doi.org/10.1111/1467-9868.03411}.

\bibitem{Yin2010}
Xiangrong Yin.
\newblock {\em Sufficient Dimension Reduction in Regression}, pages 257--273.
\newblock WORLD SCIENTIFIC / HIGHER EDUCATION PRESS, CHINA, 2010.

\bibitem{zhang_beijing_air_2017}
Shuyi Zhang, Bin Guo, Anlan Dong, Jing He, Ziping Xu, and Song~Xi Chen.
\newblock Cautionary tales on air-quality improvement in {Beijing}.
\newblock {\em Proceedings of the Royal Society A: Mathematical, Physical and
    Engineering Sciences}, 473(2205), 2017.
\newblock \url{https://doi.org/10.1098/rspa.2017.0457}.

\end{thebibliography}
\end{document}